\documentclass[12pt]{elsarticle}
\usepackage[margin=2.5cm]{geometry}
\usepackage{lipsum}

\makeatletter
\def\ps@pprintTitle{%
   \let\@oddhead\@empty
   \let\@evenhead\@empty
   \def\@oddfoot{\reset@font\hfil\thepage\hfil}
   \let\@evenfoot\@oddfoot
}
\makeatother

\usepackage{amsfonts,color,morefloats,pslatex}
\usepackage{amssymb,amsthm, amsmath,latexsym}

\newtheorem{theorem}{Theorem}
\newtheorem{lemma}[theorem]{Lemma}
\newtheorem{proposition}[theorem]{Proposition}
\newtheorem{corollary}[theorem]{Corollary}

\newtheorem{example}[theorem]{Example}

\newcommand{\rank}{{\mathrm{rank}}}

\newcommand{\tr}{{\mathrm{Tr}}}

\newcommand{\gf}{{\mathrm{GF}}}
\newcommand{\PG}{{\mathrm{PG}}}

\newcommand{\cB}{{\mathcal{B}}}

\newcommand{\C}{{\mathcal{C}}}

\newcommand{\CE}{{\mathcal{C}_{\{3,5\}}}}

\newcommand{\bD}{{\mathbb{D}}}

\newcommand{\PGL}{{\mathrm{PGL}}}

\newcommand{\tabincell}[2]{\begin{tabular}{@{}#1@{}}#2\end{tabular}}

\usepackage{blindtext}

\begin{document}

\begin{frontmatter}




\title{  The Projective General Linear Group $\PGL_2(\gf(2^m))$ and Linear Codes of Length $2^m+1$}

\tnotetext[fn1]{
C. Ding's research was supported by the Hong Kong Research Grants Council,
Proj. No. 16300418. C. Tang's research was supported by The National Natural Science Foundation of China (Grant No.
11871058) and China West Normal University (14E013, CXTD2014-4 and the Meritocracy Research
Funds).
}

\author[cding]{Cunsheng Ding}
\ead{cding@ust.hk}

\author[cmt]{Chunming Tang}
\ead{tangchunmingmath@163.com}

\author[Tonchev]{Vladimir D. Tonchev}
\ead{tonchev@mtu.edu}

\address[cding]{Department of Computer Science and Engineering, The Hong Kong University of Science and Technology, Clear Water Bay, Kowloon, Hong Kong, China}

\address[cmt]{School of Mathematics and Information, China West Normal University, Nanchong, Sichuan,  637002, China; and also  
Department of Computer Science and Engineering, The Hong Kong University of Science and Technology, Clear Water Bay, Kowloon, Hong Kong, China}

\address[Tonchev]{Department of Mathematical Sciences, Michigan Technological University, Houghton, Michigan 49931, USA}



\begin{abstract}
The projective general linear group $\PGL_2(\gf(2^m))$ 
acts as a $3$-transitive permutation group on the  set of
points of the projective line.
 The first objective of this paper is to prove that
 all linear codes over $\gf(2^h)$ that are invariant under 
$\PGL_2(\gf(2^m))$ are trivial codes:  the repetition code, the whole space 
$\gf(2^h)^{2^m+1}$, and their dual codes.
As an application of this result, the $2$-ranks of the (0,1)-incidence matrices of all
$3$-$(q+1,k,\lambda)$ designs that are invariant under  $\PGL_2(\gf(2^m))$  
are determined. 
The second objective is to present two infinite families of cyclic codes over $\gf(2^m)$ 
such that the set of the supports of all codewords of any fixed nonzero weight 
is invariant under $\PGL_2(\gf(2^m))$, therefore, the codewords
 of any nonzero weight  support a 3-design.
A code from the first  family has parameters $[q+1,q-3,4]_q$, where $q=2^m$, 
and $m\ge 4$ is even. 
The exact number of the codewords of minimum weight is determined, and
the codewords of minimum weight support a 3-$(q+1,4,2)$ 
design.
A code from the second  family has parameters $[q+1,4,q-4]_q$, $q=2^m$, $m\ge 4$ even,
and the minimum weight codewords support a 3-$(q +1,q-4,(q-4)(q-5)(q-6)/60)$ design, whose complementary 3-$(q +1, 5, 1)$ design is isomorphic
to the Witt spherical geometry  with these parameters.
A lower bound on the dimension of a linear code over  $\gf(q)$ 
 that can support
a  3-$(q +1,q-4,(q-4)(q-5)(q-6)/60)$ design
 is proved, and it is shown that the designs supported by the
 codewords of minimum weight in the codes from
the second family of codes meet this bound.

\end{abstract}


\end{frontmatter}


\section{Introduction}

A $t$-$(\nu,k,\lambda)$ design is an incidence structure $(X, \mathcal B)$, where
$X$ is a set of $\nu$ points and $\mathcal B$ a set of  $b$ $k$-subsets of $X$ called blocks,
such that any $t$ points are contained in exactly $\lambda$ blocks, where $\lambda >0$.
A $t$-design is a $t$-$(\nu,k,\lambda)$ design for some parameters $\nu,k,\lambda$.
A $t$-$(v,k,\lambda)$ design is also an $s$-$(v,k,\lambda_s)$ design for every $0\le s < t$,
where 
\[  \lambda_s =\frac{ {v-s \choose t-s}}{{k-s ] \choose t-s}}\lambda. \]
In particular, the number of blocks is equal to
\[ b=\lambda_0 = \frac{{v \choose t}}{{ k \choose t}}\lambda. \]

The incidence matrix $A=(a_{i,j})$ of a design $\bD$ is a (0,1)-matrix with rows
indexed by the blocks, and columns indexed by the points of $\bD$, where $a_{i,j}=1$
if the $j$th point belongs to the $i$th block, and $a_{i,j}=0$ otherwise.
If $q$ is a prime power, the $q$-rank of $\bD$ (or $\rank_{q}\bD$) is defined as
the rank of its incidence matrix $A$ over a finite field $\gf(q)$ of order $q$:
 $\rank_{q}\bD = \rank_{q}A$.
Equivalently, the $q$-rank of a design is the dimension of the linear $q$-ary code
spanned by the rows of its (0,1)-incidence matrix.

A {\it generalized } incidence matrix of a design $\bD$ over a finite field $\gf(q)$,
or shortly, an  $\gf(q)$-incidence matrix of $\bD$, is any matrix obtained by replacing
the nonzero entries of the (0,1)-incidence matrix of $\bD$ with arbitrary nonzero
elements of $\gf(q)$. The {\it dimension} of a $t$-$(v,k,\lambda)$
design $\bD$ over $\gf(q)$ (or the $q$-dimension of $\bD$, or $\dim_{q}\bD$),
is defined in \cite{Tdim} as the minimum among the dimensions of all linear codes
of length $v$ over $\gf(q)$ that contain the blocks of $\bD$ among the supports of codewords 
of weight $w$. Equivalently, the $q$-dimension of $\bD$ is equal to 
\[ \dim_{q}\bD = \min \ \rank_{q}M, \]
where $M$ runs over the set of all $(q-1)^{bk}$ generalized  $\gf(q)$-incidence
matrices of $\bD$, and $b$ is the number of blocks.
Clearly, $\dim_{q}\bD \le \rank_{q}\bD$. For example, if $\bD$ is the 4-$(11,5,1)$ design
supported by the codewords of minimum weight in the ternary Golay code of length 11 and dimension 6,
$\dim_{3}\bD=6$, while $\rank_{3}\bD=11$.
A generalization of this definition  of the $q$-dimension of a design 
 is given in \cite{JT13}.

The importance of interactions between groups, linear codes and $t$-designs has been 
well recognized for decades.  
For example, Assmus and Mattson  \cite{AM69} pointed out in 1969 that
 $5$-designs arise from certain extremal self-dual codes, including the
extended Golay codes that are closely  related to the 5-transitive Mathieu groups. 
Linear codes that are invariant under  groups acting on the  set of code coordinates  
have found important applications  for the construction of combinatorial $t$-designs. 
Examples of such codes are the Golay codes, the quadratic-residue codes,
and the affine-invariant codes \cite[Chapter 6]{Dingbook18}. 
  
  This paper presents a number of new results about $3$-designs arising from linear codes 
associated  with the projective general linear group $\PGL_2(\gf(2^m))$.

Let $\PGL_2(\gf(q))$ be the projective general linear group acting
as a permuttion group on the set of points of the projective line $\PG(1,q)$ 
over a finite field $\gf(q)$ with $q$ elements.
Every vector in the $(q+1)$-dimensional vector space $\gf (r)^{q+1}$  can be written as 
$(c_{x})_{x\in \PG(1,q)}$, where $c_x \in \gf(r)$ and $r$ is a prime power. In other words, the coordinates of the vectors in $\gf (r)^{q+1}$ can be indexed by the points  
in $\PG(1,q)$.  
Consider the induced action of $\PGL_2(\gf(q))$ on $\gf (r)^{q+1}$ by the left translation:
$$ 
\pi:  (c_{x})_{x\in \PG(1,q)}  \mapsto (c_{\pi (x)})_{x\in \PG(1,q)}, 
$$ 
where $(c_{x})_{x\in \PG(1,q)} \in \gf(r)^{q+1}$ and $\pi \in  \PGL_2(\gf(q))$.
Let $\mathcal C$ be a linear code of length $q+1$ over $\gf(r)$.
We say that  $\mathcal C$  is 
\emph{invariant under $\PGL_2(\gf(q))$} if each element of $\PGL_2(\gf(q))$ 
carries each codeword of $\mathcal C$ into a codeword of $\mathcal C$.
In other words, $\mathcal C$ is invariant under $\PGL_2(\gf(q))$ 
 if $\mathcal C$ admits $\PGL_2(\gf(q))$ as a subgroup of the permutation automorphism group of $\C$.
 For a codeword $\mathbf c =(c_x)_{x\in \PG(1,q)}$ in $\mathcal C$, the \emph{support} of  
 $\mathbf c$
is defined as
\begin{align*}
\mathrm{Supp}(\mathbf c) = \{x \in \PG(1,q): c_x \neq 0 \}.
\end{align*}
Let $A_{w}(\mathcal  C)= | \left \{\mathbf c \in \mathcal C: wt(\mathbf{c})=w \right \}|$ and
$\mathcal B_{w}(\mathcal C)
=\{   \mathrm{Supp}(\mathbf c): wt(\mathbf{c})=w
~\text{and}~\mathbf{c}\in \mathcal{C}\}$, where $wt(\mathbf{c})$ denotes the Hamming weight of $\mathbf c$.
$\mathcal B_{w}(\mathcal C)$ is said to be invariant under $\PGL_2(\gf(q))$ if the support
$\mathrm{Supp} \left ( (c_{\pi(x)})_{x\in \PG(1,q)} \right )$ belongs to 
$\mathcal B_{w}(\mathcal C)$
for every $\pi \in \PGL_2(\gf(q))$ and any codeword $(c_{x})_{x\in \PG(1,q)}$ of weight $w$ in 
$\mathcal C$.
It is easily seen that if $\mathcal C$ is invariant under $\PGL_2(\gf(q))$, then so is 
$\mathcal B_{w}(\mathcal C)$ for each $w$.
Moreover, if $\mathcal B_{w}(\mathcal C)$ is invariant under $\PGL_2(\gf(q))$, then 
$\left ( \PG(1,q),  \mathcal B_{w}(\mathcal C) \right)$
holds a $3$-design provided $A_w  (\mathcal C)\neq 0$,
since the action of $\PGL_2(\gf(q))$ on $\PG(1,q)$ is $3$-transitive (see 
\cite[Propositions 4.6 and 4.8]{BJL} or \cite[Proposition 1.27]{Tonchevhb}).
For more related results on linear codes and $t$-designs, 
we refer the reader to \cite{DWF,Dingbook18}. 

The first objective of this paper is to investigate the possible parameters of  linear codes
that are invariant under  $\PGL_2(\gf(q))$. 
We focus on the case when $q$ and $r$ are powers of $2$.
We prove in  Section 4, Theorem \ref{thm:code-PGL-4}, that the only linear codes of length 
$2^m+1$ over $\gf(2^h)$ that
 are invariant under $\PGL_2(\gf(2^m))$  are  trivial codes: the zero code, the whole space 
 $\gf(2^h)^{2^m+1}$, the repetition code, and its dual code. 
 As an application of this result, the $2$-ranks of the (0,1)-incidence matrices of all
$3$-$(q+1,k,\lambda)$ designs that are invariant under  $\PGL_2(\gf(2^m))$  
are determined,   and it is proved in Theorem \ref{thm;2-rank}
that any such design has  2-rank equal to $q+1$ if the block size $k$ is odd,
and  $q$ if $k$ is even.
 
The second objective of this paper  is to investigate the question whether
there are  any nontrivial  linear codes of length $2^m+1$ over  $\gf(2^m)$, such that
the set of  the supports of all codewords of any fixed  nonzero weight is invariant under  
$\PGL_2(\gf(2^m))$. In Section 5, we answer this question in the affirmative 
by presenting two infinite families of cyclic codes 
of length $2^m+1$ over $\gf(2^m)$, such that the set of the supports of the codewords of any fixed weight is invariant  under $\PGL_{2}(\gf(2^m))$, therefore, the codewords
 of any nonzero weight  support a 3-design.
These codes are obtained as subfield subcodes and trace codes of certain cyclic codes
 over $\gf(2^{2m})$ and their dual codes (Theorems \ref{thm:dual-C-3-5} and
 \ref{thm:C-3-5}).

A code from the first  family has parameters $[q+1,q-3,4]_q$, where $q=2^m$, 
and $m\ge 4$ is even. 
The exact number of the codewords of minimum weight is determined, and
the codewords of minimum weight support a 3-$(q+1,4,2)$ 
design.
To the best  knowledge of the authors, this is the first infinite family of 
linear codes that support an infinite family of $3$-$(v, 4, 2)$ designs.
 The codewords of every other nonzero weight also support  3-designs.
 
 A code from the second  family has parameters $[q+1,4,q-4]$, $q=2^m$, 
 $m\ge 4$ even.
 The exact number of the codewords of minimum weight is determined, and
 the minimum weight codewords support a 3-$(2^m +1,q-4,\lambda)$ design   
 with
 \[ \lambda=\frac{(q-4)(q-5)(q-6)}{60}, \]
  whose complementary 3-$(q +1, 5, 1)$ design is shown to be isomorphic
to the Witt spherical geometry with these parameters. 
 In Section \ref{Sec6}, a lower bound on the $q$-dimension of
a 3-$(q+1,(q-4),(q-4)(q-5)(q-6)/60)$ design
 is proved in Theorem \ref{bound}, and it  is shown
that  the infinite family of 3-designs  described in Theorem \ref{thm:10-5} meet this bound.


\section{Preliminaries}

\subsection{Group actions and $t$-designs}

A \emph{permutation group} is a subgroup of the \emph{symmetric group} $\mathrm{Sym}(X)$, where $X$ is a finite set.
More generally, an \emph{action} $\sigma$ of a finite group $G$ on a set $X$ is a homomorphism $\sigma$ from $G$ to $\mathrm{Sym}(X)$.
We denote the image $\sigma(g)(x)$ of $x\in X$ under $g\in G$ by $g(x)$ when no confusion can arise.
The \emph{$G$-orbit} of $x\in X$ is $\mathrm{Orb}_{x}=\{g(x): g \in G\}$.
The \emph{stabilizer} of $x$ is $\mathrm{Stab}_x=\{g \in G: g(x)=x\}$. The length of the 
orbit of $x$ is given by
\[ \left | \mathrm{Orb}_{x} \right | = \left | G \right | / \left | \mathrm{Stab}_x \right |.\]

One criterion to measure the level of symmetry is the \emph{degree} of \emph{transitivity} 
and \emph{homogeneousity} of the group. 
Recall that a group $G$ acting on a set $X$ is  \emph{$t$-transitive} (resp., 
\emph{$t$-homogeneous}) if
for any two ordered $t$-tuples $(x_1, \cdots, x_t), (x_1', \cdots, x_t')$
of  distinct elements from $X$ (resp., two unordered $t$-subsets 
$\{x_1, \cdots, x_t\}$, $\{x_1', \cdots, x_t'\}$ of $X$) there is some $g\in G$ such that 
$\left (x_1', \cdots, x_t' \right )= \left (g(x_1), \cdots, g(x_t) \right )$
(resp., $\left \{x_1', \cdots, x_t' \right \}=\left \{g(x_1), \cdots, g(x_t) \right \}$).

We recall a well-known general fact (see, e.g.  \cite[Proposition 4.6]{BJL}), that  for a $t$-homogeneous group $G$ on a finite set $X$ with
$ |X|=\nu$  and a subset $B$ of $X$ with $|B| =k >t$, the pair 
$(X, \mathrm{Orb}_{B})$ is a 
$t$-$(\nu, k , \lambda)$ design,
where $\mathrm{Orb}_{B}$ is the set of images of $B$ under
the group $G$, 
$\lambda=\frac{\binom{k}{t} |G| }{\binom{\nu}{t} | \mathrm{Stab}_{B}|}$ and 
$\mathrm{Stab}_{B}$ is the setwise stabilizer of
$B$ in $X$. Let $\binom{X}{k}$ be the set of subsets of $X$ consisting of $k$ elements.
A nonempty subset $\mathcal B$ of $\binom{X}{k}$ is called \emph{invariant} under
 $G$ if $\mathrm{Orb}_{B} \subseteq \mathcal B$ for any $B \in \mathcal B$.
If this is the case, it means that the pair $(X, \mathcal B)$ is a $t$-$(\nu, k , \lambda)$ design
admitting $G$ as an automorphism group for some $\lambda$. For some recent works on  $t$-designs from group actions,  
we refer the reader to \cite{Tang,XLW}.

\subsection{Projective general linear groups of degree two}
The \emph{projective linear group $\PGL_2(\gf(q))$ of degree two} is defined as the group
of invertible $2\times 2$ matrices with entries in $\gf(q)$,
 modulo the scalar matrices, 
 $\begin{bmatrix}
a & 0\\
0 &  a
\end{bmatrix}$, where $a\in \gf(q)^*$.
Note that the group $\PGL_2(\gf(q))$ is generated by the matrices
$\begin{bmatrix}
a & 0\\
0 &  1
\end{bmatrix}$,
$\begin{bmatrix}
1 & b\\
0 &  1
\end{bmatrix}$
and
$\begin{bmatrix}
0 & 1\\
1 &  0
\end{bmatrix}$, where $a\in \gf(q)^*$ and $b\in \gf(q)$.

Here the following convention for the action of $\PGL_2(\gf(q))$ on the projective line 
$\mathrm{PG}(1,q)$
is used. A matrix 
$\begin{bmatrix}
a & b\\
c &  d
\end{bmatrix}
\in \PGL_2(\gf(q))$ acts on $\mathrm{PG}(1,q)$ by
\begin{eqnarray}\label{eq:action-PGL(2,q)}
\begin{array}{c}
(x_0 : x_1) \mapsto \begin{bmatrix}
a & b\\
c &  d
\end{bmatrix} (x_0 : x_1) =   (a x_0 +b x_1 : c x_0 +d x_1),
\end{array}
\end{eqnarray}
or, via the usual identification of $\gf(q) \cup \{ \infty\}$ with  $\mathrm{PG}(1,q)$, by linear fractional transformation
\begin{eqnarray}\label{eq:action-infty}
\begin{array}{c}
 x \mapsto \frac{a x +b  }{c x +d}.
\end{array}
\end{eqnarray}
This is an action on the left, i.e., for $\pi_1, \pi_2 \in  \PGL_2(\gf(q))$
and $x \in  \PG(1,q)$ the following holds: $\pi_1 (\pi_2(x)) = (\pi_1 \pi_2)(x)$.
The action of $\PGL_2(\gf(q))$ on $\PG(1,q)$ defined in (\ref{eq:action-infty}) 
is sharply $3$-transitive, i.e.,
for any distinct $a, b, c \in \gf(q) \cup \{\infty\}$ there is $\pi \in \PGL_2(\gf(q))$
taking $\infty$ to $a$, $0$ to $b$,  and $1$ to $c$.
 In fact, $\pi$ is uniquely determined and it equals
\[ \pi= \begin{bmatrix}
 a(b-c) & b(c-a)\\
 b-c & c-a
 \end{bmatrix}.
\]
Thus, $\PGL_2(\gf(q))$ is in one-to-one correspondence with the set of ordered triples 
$(a,b,c)$ of distinct elements in $\gf(q) \cup \{\infty\}$, and in particular
\begin{eqnarray}\label{eq:cardinality of PGL(2,q)}
\begin{array}{c}
| \PGL_2(\gf(q))| = (q+1)q(q-1).
\end{array}
\end{eqnarray}

Two subgroups $H_1$ and $H_2$ of a group $G$ are said to be \emph{conjugate} if there is a $g \in G$ such 
that $gH_1g^{-1}=H_2$. It is easily seen that this conjugate relation is an equivalence relation on the set of all subgroups 
of $G$, and is called the conjugacy.  The conjugacy classification of subgroups of 
$\PGL_2(\gf(2^m))$ is well known \cite{Dickson01}. Table \ref{tab:subgroups}
 specifies all the subgroups of $\PGL_2(\gf(2^m))$ up to conjugacy.

 \begin{table}[!htbp]
\centering
\begin{tabular}{|c|c|c|c|}
\hline
Type & Maximal order & \tabincell{c}{Number of \\ conjugacy classes} & Condition\\
\hline
\hline
$2$-group & $2^m$ &  $-$ & $-$\\
\hline
Frobenius & $2^m(2^m-1)$ &  $-$ & $-$\\
\hline
Cyclic & $2^m-1$ &  one & $-$\\
\hline
Cyclic & $2^m+1$ &  one & $-$\\
\hline
Dihedral & $2(2^m-1)$ &  one  & $-$ \\
\hline
Dihedral & $2(2^m+1)$ &  one  & $-$ \\
\hline
$\PGL_2(\gf(2^{m'}))$ & $2^{m'}(4^{m'}-1)$ &  one & $m'| m$\\
\hline
$A_4$ & 12 &  $-$  & $2|m$\\
\hline
 $A_5$ & 60 &  $-$ &  $2|m$\\
\hline
\end{tabular}
\caption{Subgroups of $\PGL_2(\gf(2^m))$}\label{tab:subgroups}
\end{table}

We recall here the classification of sharply $3$-transitive finite permutation groups on finite sets of odd cardinality (see for instance \cite{Passman68}).

\begin{theorem}\label{thm:sharply-3-transitive}
Let $G$ be a sharply $3$-transitive permutation group on the finite set $X$ of odd cardinality.
Then it is possible to identify the elements of $X$ with the points of the projective line 
$\PG(1,2^m)$
in such a way that $G = \PGL_2(\gf(2^m))$ holds.
\end{theorem}

\subsection{Linear codes and cyclic codes}

Let $\gf(r)$ be the finite field with $r$ elements.
An \emph{$[n,k]_{r}$ linear code} $\mathcal C$ is a $k$-dimensional vector subspace of $\gf(r)^{n}$.
If it has minimum distance $d$ it is also called an $[n, k, d]_r$ code.
The \emph{dual code} $\mathcal C^{\perp}$ of $\mathcal C$ is the set of vectors orthogonal to all codewords of $\mathcal C$:
\[ \mathcal C^{\perp} = \{\mathbf w \in \gf(r)^{n} :  \langle \mathbf c, \mathbf w \rangle =0 \text{ for all } \mathbf c \in \mathcal C\}, \]
where $ \langle \mathbf c, \mathbf w \rangle$ is the usual Euclidean inner product of $\mathbf c$ and $ \mathbf w $.
Let $\mathbf{a}=(a_0, \cdots, a_{n-1}) \in \left ( \gf (r)^* \right )^n$.
Here and subsequently, $\mathbf{a} \cdot \mathcal C$ stands for the linear code
$\left \{ (a_0 c_0, \cdots, a_{n-1}c_{n-1}): ( c_0, \cdots, c_{n-1}) \in \mathcal C \right \}$.
It is a simple matter to check that
\begin{eqnarray}\label{eq:a-code-dual}
\left ( \mathbf{a} \cdot \mathcal C \right )^{\perp} =  \mathbf{a}^{-1} \cdot  \mathcal C ^{\perp},
\end{eqnarray}
where $\mathbf{a}^{-1}=(a_0^{-1}, \cdots, a_{n-1}^{-1})$.

There are two classical ways to construct a code over $\gf(r)$ from a given code over $\gf(r^h)$.
Let $\mathcal C$ be a code of length $n$ over $\gf(r^h)$. Then the \emph{subfield subcode }
$\mathcal C |_{\gf(r)}$ equals $\mathcal C \cap \gf(r)^n  $, the set of those codewords of $\mathcal C$ all of whose coordinate entries belong to the subfield $\gf(r)$.
The \emph{trace code} of $\mathcal C$ is given by
\[  \tr_{r^h/r} (\mathcal C) =\left \{ \left (\tr_{r^h/r}(c_0), \cdots,  \tr_{r^h/r}(c_{n-1})\right ): (c_0, \cdots, c_{n-1}) \in \mathcal C  \right \}, \]
where $\tr_{r^h/r}$ denotes the trace function from  $\gf(r^h)$ to $\gf(r)$. A celebrated result of Delsarte \cite{Del75} states
that the subfield code $ \left . \mathcal C^{\perp} \right |_{\gf(r)}$ and the trace code $ \tr_{r^h/r} (\mathcal C)$ are duals of each other, namely, 
\begin{eqnarray}\label{eq:subfield-trace}
\left ( \tr_{r^h/r} (\mathcal C) \right )^{\perp}= \left . \mathcal C^{\perp} \right |_{\gf(r)}.
\end{eqnarray}

Conversely, given a linear code $\mathcal  C$ of length $n$ and dimension $k$ over $\gf(r)$, we define
 a linear code $  \gf(r^h) \otimes \mathcal C$ over $\gf(r^h)$ by 
\begin{eqnarray}
 \gf(r^h) \otimes \mathcal C=\left \{\sum_{i=1}^{k} a_i \mathbf{c}_i: (a_1, a_2, \ldots, a_k) \in \gf(r^h)^k \right \},
\end{eqnarray} 
where $\left \{ \mathbf{c}_1, \mathbf{c}_2, \ldots, \mathbf{c}_k \right \}$ is a basis of $\mathcal C$ over $\gf(r)$. 
This code is independent of the choice of the basis $\left \{ \mathbf{c}_1, \mathbf{c}_2, \ldots, \mathbf{c}_k \right \}$ 
of $\C$, is called the \emph{lifted code} of $\mathcal C$ to $\gf(r^h)$. Clearly, $\gf(r^h) \otimes \mathcal C$ 
and $\C$ have the same length, dimension and minimum distance, but different weight distributions.  
A trivial verification shows that if $(c_0, \cdots, c_{n-1}) \in \gf(r^h) \otimes \mathcal C $,
then  $(c_0^r, \cdots, c_{n-1}^r) \in \gf(r^h) \otimes \mathcal C $. Applying  \cite[Lemma 7]{GP10}, one has
\[  \tr_{r^h/r} \left (  \gf(r^h) \otimes \mathcal C \right ) =\left . \left (  \gf(r^h) \otimes \mathcal C  \right ) \right |_{\gf(r)}. \]

Let $n$ be a positive integer with $\gcd(n,r)=1$.
The \emph{order} $\mathrm{ord}_{n}(r)$ of $r$ modulo $n$ is the smallest positive integer $h$ such that $r^h \equiv  1 \pmod{n}$.
Let $\mathbb Z_n$ denote the ring of residue classes of integers modulo $n$.
The $r$-cyclotomic coset of $e\in \mathbb Z_n$ is the set $[e]_{(r,n)}= \{ r^ie: 0 \le i \le \mathrm{ord}_n(r)-1\}$.
Then any two $r$-cyclotomic cosets are either equal or disjoint.
A subset $E$ of $\mathbb Z_n$ is called $r$-invariant if
the set $\{re: e \in E\}$ equals $E$, that is, $E$ is the union of some $r$-cyclotomic cosets.
A subset $\widetilde{E}=\{e_1, \cdots, e_t\}$ of an $r$-invariant set $E$ is called a complete set of representatives of $r$-cyclotomic cosets of $E$
if $[e_1]_{(r,n)}, \cdots, [e_t]_{(r,n)}$  are pairwise distinct and $E= \cup_{i=1}^{t} [e_i]_{(r,n)}$.

An $[n, k]_r$ code  $\mathcal C$ is \emph{cyclic} if $(c_0, c_1, \cdots, c_{n-1}) \in \mathcal  C$ implies that $(c_{n-1}, c_0, ..., c_{n-2}) \in  \mathcal C$.
Let $\gamma$ be a primitive $n$-th root of unity in $\gf(r^h)$, where $h=\mathrm{ord}_{n}(r)$.
It is known \cite{HP03} that any $r$-ary cyclic code of length $n$ with $\gcd(n,r)=1$ has a simple description by means of the trace function.

\begin{theorem}\label{thm:cyclic-trance}
Let $\mathcal C$ be an $[n,k]_r$ cyclic code with $\gcd(n,r)=1$ and $\gamma$ be a primitive
$n$-th root of unity in $\gf(r^h)$, where $h=\mathrm{ord}_{n}(r)$.
Then there exists a unique $r$-invariant set $E \subseteq \mathbb Z_n$ such that
\begin{eqnarray*}
\mathcal C= \left \{  \left (  \sum_{i=1}^{t} \tr_{r^{h_i}/r} \left ( a_i \gamma^{e_ij} \right ) \right )_{j=0}^{n-1}: a_i \in \gf\left (r^{h_i} \right )\right \},
\end{eqnarray*}
where $\{e_1, \cdots, e_t\}$ is any complete set of representatives of $r$-cyclotomic cosets of $E$ and $h_i=| [e_i]_{(r,n)}|$.
Moreover, $k=|E| = \sum_{i=1}^t h_i$.
\end{theorem}

Theorem \ref{thm:cyclic-trance} states that there is a one-to-one correspondence between cyclic linear codes
over $\gf(r)$ with length $n$ and $r$-invariant subsets of 
$\mathbb Z_n$ with respect to a fixed $n$-th root of 
unity $\gamma$.
We will call the set $E$ in Theorem \ref{thm:cyclic-trance}  the \emph{cyclicity-defining set} of $\mathcal C$ with respect to $\gamma$.

The following corollary is an immediate consequence of Theorem \ref{thm:cyclic-trance}.

\begin{corollary}\label{cor:cyclic-extension}
Let $n$ be a positive integer such that $\gcd(n,r)=1$.
 Let $\mathcal C$ be an $[n,k]_r$ cyclic code with cyclicity-defining set $E$ and $\gf(r^{\ell}) \otimes \mathcal C$ be the lifted code of $\C$ to $\gf(r^{\ell})$.
Then $\gf(r^{\ell}) \otimes \mathcal C$ is an $[n,k]_{r^{\ell}}$ cyclic code defined by the cyclicity-defining set $E$ of $\mathcal C$. In particular, 
\begin{eqnarray*}
\gf(r^{h}) \otimes  \mathcal C= \left \{  \left (  \sum_{e\in E} a_e \gamma^{je} \right )_{j=0}^{n-1}: a_e \in \gf\left (r^{h} \right )\right \},
\end{eqnarray*}
where $h=\mathrm{ord}_n(r)$ and $\gamma$ is a primitive $n$-th root of unity in $\gf(r^h)$.
\end{corollary}

Since the set $E$ also defines the code $\gf(r^{\ell}) \otimes \mathcal C$ in Corollary \ref{cor:cyclic-extension}, 
the set $E$ is also called the cyclicity-defining set of the lifted code  $\gf(r^{\ell}) \otimes \mathcal C$.

Let $n$ be a positive integer with $\gcd(n,r)=1$ and $h=\mathrm{ord}_n(r)$.
 Let $U_n$ be the cyclic multiplicative group of all $n$-th roots of unity in $\gf(r^h)$.
By polynomial interpolation, every function $f$ from $U_{n}$ to $\gf(r)$  has a unique \emph{univariate polynomial expansion} of the form
\[ f(u)= \sum_{i=0}^{n-1} a_i u^i,\]
where $a_j \in \gf(r^h)$, $u \in U_n$.

As a direct result of Theorem \ref{thm:cyclic-trance}, we have the following conclusion concerning cyclicity-defining sets of cyclic codes.

\begin{corollary}\label{cor:a_e-neq-0-E}
Let $n$ be a positive integer with $\gcd(n,r)=1$, $h=\mathrm{ord}_n(r)$ and $\gamma$ a primitive $n$-th root of unity in $\gf(r^h)$.
Let $\mathcal C$ be an $[n,k]_r$ cyclic code with cyclicity-defining set $E$. Let $f(u)=\sum_{i=0}^{n-1} a_i u^i \in \gf(r^h) [u]$.
If $\left (f(\gamma^j) \right )_{j=0}^{n-1} \in \mathcal C$ and $a_i\neq 0$, then $i\in E$.
\end{corollary}

\section{Another representation of the action of $\PGL_2(\gf(2^m))$ on the projective line $\PG(1,2^m)$}

In this section we give another representation of the action of $\PGL_2(\gf(2^m))$ on the projective line $\PG(1,2^m)$.
This new representation will play an important role in Sections \ref{sec:codes-PGL} and \ref{sec:supports-PGL}.

Let $U_{q+1}$  be the subset of the projective line $\PG(1,q^2)=\gf(q^2) \cup \{\infty\}$ consisting of all the $(q+1)$-th roots of unity. Denote by $\mathrm{Stab}_{U_{q+1}}$ the setwise stabilizer of $U_{q+1}$ under the action of 
$\PGL_2(\gf(q^2))$ on $\PG(1,q^2)$.

\begin{proposition}\label{prop:three-types}
Let $q=2^m$. Then the setwise stabilizer $\mathrm{Stab}_{U_{q+1}}$ of $U_{q+1}$ consists of
the following three types of linear fractional transformations:
\begin{enumerate}
\item[(I)] $u \mapsto u_0 u$, where $u_0 \in U_{q+1}$;
\item[(II)] $u \mapsto u_0 u^{-1}$, where $u_0 \in U_{q+1}$;
\item[(III)] $u \mapsto \frac{u+c^q u_0}{c u +u_0}$, where $u_0 \in U_{q+1}$ and $c \in \gf(q^2)^* \setminus U_{q+1}$.
\end{enumerate}

\end{proposition}

\begin{proof}
First, the transformations listed in (I)-(III)  are easily seen to belong to the stabilizer $\mathrm{Stab}_{U_{q+1}}$.

Conversely, let $\pi$ be a translation in $ \PGL_2(\gf(q^2))$ given by $\frac{a x+b}{cx+d}$, where $a,b,c,d \in \gf(q^2)$ and $ad+bc\neq 0$.
Then $\pi \in \mathrm{Stab}_{U_{q+1}}$ if and only if the following holds
\begin{eqnarray}\label{eq:(q+1)/(q+1)}
\left ( \frac{a u+b}{cu+d} \right )^{q+1}=1, \text{ for all } u \in U_{q+1}.
\end{eqnarray}
Multiplying both sides of (\ref{eq:(q+1)/(q+1)}) by $(cu+d)^{q+1}$ yields
\begin{eqnarray*}
(a^{q+1}+c^{q+1})u^{q+1}+(a^qb+c^qd)u^q+(ab^q+cd^q)u+(b^{q+1}+d^{q+1})=0.
\end{eqnarray*}
Substituting $u^{-1}$  for $u^q$ in the equation above yields 
\begin{eqnarray}\label{eq:quadratic-eq-3}
(ab^q+cd^q)u^2+(a^{q+1}+b^{q+1}+c^{q+1}+d^{q+1})u+(a^qb+c^qd)=0.
\end{eqnarray}
Since the quadratic equation in (\ref{eq:quadratic-eq-3}) has at least $q+1$ roots: $u\in U_{q+1}$, all its coefficients must be zero, that is
\begin{eqnarray}\label{eq:degree--q+1}
\left \{
\begin{array}{cl}
 ab^q+cd^q & =0, \\
 a^qb+c^qd & =0, \\
 a^{q+1}+b^{q+1}+c^{q+1}+d^{q+1} &=0. 
\end{array}
\right . 
\end{eqnarray}

We investigate the following three cases for (\ref{eq:degree--q+1}).

If $b=0$, (\ref{eq:degree--q+1}) clearly forces $c=0$. Thus $\pi =u_0 x$ for some $u_0 \in U_{q+1}$.

If $a=0$,  (\ref{eq:degree--q+1}) clearly forces $d=0$. Thus $\pi =u_0 x^{-1}$ for some $u_0 \in U_{q+1}$.

If $ab \neq 0$, we can certainly assume that $a=1$, because $\frac{ax+b}{cx+d}$ and $\frac{x+a^{-1}b}{a^{-1}cx+a^{-1}d}$
determine the same translation. Substituting $1$ for $a$ in (\ref{eq:degree--q+1}) we conclude that
\begin{eqnarray*}
\left \{
\begin{array}{cl}
 b+c^qd & =0\\
 1+b^{q+1}+c^{q+1}+d^{q+1} &=0
\end{array}
\right . .
\end{eqnarray*}
This gives $b=c^q d$ and $(c^{q+1}+1) (c^{q+1}+1)=0$.
If $c^{q+1}=1$, we would have $ad+bc=0$, a contradiction.
 It follows that
$b=c^q d$, $d^{q+1}=1$ and $c\in \gf(q^2)^* \setminus U_{q+1}$ from $ad+bc \neq 0$.
This completes the proof.

\end{proof}

The following result follows from Proposition \ref{prop:three-types} directly.

\begin{corollary}\label{cor:generators-three-types}
Let $q=2^m$. Then the setwise stabilizer $\mathrm{Stab}_{U_{q+1}}$ of $U_{q+1}$ is generated by
the following three types of linear fractional transformations:
\begin{enumerate}
\item[(I)] $u \mapsto u_0 u$, where $u_0 \in U_{q+1}$;
\item[(II)] $u \mapsto u^{-1}$;
\item[(III)] $u \mapsto \frac{u+c^q }{c u +1}$, where  $c \in \gf(q^2)^* \setminus U_{q+1}$.
\end{enumerate}

\end{corollary}

The following proposition shows that the action of  $\mathrm{Stab}_{U_{q+1}}$ on $U_{q+1}$
and the action of $\PGL_2(\gf(2^m))$ on $\PG(1,2^m)$ are equivalent.

\begin{proposition}\label{prop:Stab-PGL}
Let $q=2^m$ and $\mathrm{Stab}_{U_{q+1}}$   the setwise stabilizer of $U_{q+1}$.
Then $\mathrm{Stab}_{U_{q+1}}$  is conjugate in $\PGL_2(\gf(2^{2m}))$ to the group $\PGL_2(\gf(2^m))$,
and its action on $U_{q+1}$ is equivalent to the action of $\PGL_2(\gf(2^m))$ on $\PG(1,2^m)$.

\end{proposition}

\begin{proof}
We begin by proving
that the group $\mathrm{Stab}_{U_{q+1}}$ acts sharply $3$-transitively on $U_{q+1}$.
It suffices to show that this action is $3$-transitive as $\PGL_2(\gf(2^{2m}))$ acts sharply $3$-transitively on $\PG(1,2^{2m})$.
Let $(u_1, u_2, u_3), (u_1', u_2', u_3')$ be any two $3$-tuples
of distinct elements from $U_{q+1}$. Since the action of  $\PGL_2(\gf(2^{2m}))$ on
$\PG(1,2^{2m})$ is $3$-transitive, there exists a linear fractional transformation $\frac{a x +b}{cx +d} \in \PGL_2(\gf(2^{2m}))$
such that
\[\frac{a u_i +b}{cu_i +d}= u_i' \text{ for } i=1,2,3.\]
This gives
\[ \left ( \frac{a u_i +b}{cu_i +d} \right )^{q+1}= 1, \text{ where } i=1,2,3.\]

Using a similar argument to Proposition \ref{prop:three-types}, we can prove that
the transformation given by $u \mapsto \frac{au+b}{cu+d}$
belongs to $\mathrm{Stab}_{U_{q+1}}$, which says that the action of $\mathrm{Stab}_{U_{q+1}}$ on $U_{q+1}$ is $3$-transitive.
Therefore the action of $\mathrm{Stab}_{U_{q+1}}$ on $U_{q+1}$ is equivalent to
the action of $\PGL_2(\gf(2^m))$ on $\PG(1,2^m)$ by Theorem \ref{thm:sharply-3-transitive}.
Now Table \ref{tab:subgroups} shows that $\mathrm{Stab}_{U_{q+1}}$  is conjugate to the subgroup $\PGL_2(\gf(2^m))$ in $\PGL_2(\gf(2^{2m}))$.
This completes the proof.
\end{proof}

\section{Linear codes  invariant under  $\PGL_2(\gf(2^m))$}\label{sec:codes-PGL}

The main objective of this section is to classify all linear codes over $\gf(2^h)$ of length $2^m+1$ that are invariant under 
$\PGL_2(\gf(2^m))$.
As an immediate application,
 we derive
the $2$-rank of the incidence matrices of $t$-$(2^m+1, k, \lambda)$ designs that are invariant under $\PGL_2(\gf(2^m))$.

Let $\mathcal C$ be a $[2^m+1, k]_{2^h}$ linear code.
We can regard $U_{2^m+1}$ as the set of the coordinate positions of $\mathcal C$ and write the codeword of $\mathcal C$ as $\left ( c_u\right )_{u\in U_{2^m+1}}$.
Then the set of coordinate positions of $\mathcal C$ could be endowed with the
action of $\mathrm{Stab}_{U_{2^m+1}}$. According to Proposition 
\ref{prop:Stab-PGL},
we only need to find  all linear codes over $\gf(2^h)$ of length $2^m+1$ which are invariant under $\mathrm{Stab}_{U_{2^m+1}}$.

The following lemma gives the polynomial expansion of the linear fractional transformation 
$\frac{u+c^q}{cu+1}$, where $c \in \gf(q^2)^* \setminus U_{q+1}$.
\begin{lemma}\label{lem:frac-poly}
Let $q=2^m$ with $m \ge 2$  and $c \in \gf(q^2)^* \setminus U_{q+1}$.
Then for any $u\in U_{q+1}$,  the following holds
\begin{eqnarray*}
\frac{u+c^q}{cu+1} =\sum_{i=1}^{q} c^{i-1} u^i.
\end{eqnarray*}

\end{lemma}
\begin{proof}
An easy computation shows that
\begin{eqnarray*}
\begin{array}{rl}
\sum_{i=1}^{q} c^{i-1} u^i & = \frac{1+(cu)^{q}}{1+cu} u\\
 & = \frac{u+c^qu^{q+1}}{1+cu} \\
 &= \frac{u+c^q}{cu+1},
\end{array}
\end{eqnarray*}
which completes the proof.

\end{proof}

The following lemma expresses the coefficients of the polynomial expansion of a function $f$ over $U_{q+1}$
in terms of the sums  over $U_{q+1}$ of the product function of $f$ and  the power functions $u^j$.
\begin{lemma}\label{lem:ai=sum-f}
Let $f$ be a function from $U_{q+1}$ to $\gf(q^{2h})$ with $h\ge 1$. Let 
$\sum_{i=0}^{q} a_i u^i$
be the polynomial expansion of $f$, where $a_i \in \gf(q^{2h})$. Then 
$a_i=\sum_{u \in U_{q+1}} f(u) u^{-i}$, where $0 \le i \le q$.
\end{lemma}

\begin{proof}
A straightforward computation yields that
\begin{eqnarray}\label{eq:power-sum}
\sum_{u\in U_{q+1}} u^e =
\left \{
\begin{array}{rc}
1, & \text{ if } q+1 \text{ divides } e,\\
0, & \text{ otherwise},
\end{array}
\right .
\end{eqnarray}
where $e$ is an integer.

A standard calculation shows that
\begin{eqnarray*}
\begin{array}{rl}
\sum_{u \in U_{q+1}} f(u) u^{-i}&=  \sum_{u\in U_{q+1}} u^{-i}\sum_{j=0}^q a_j u^j\\
&= \sum_{j=0}^q a_j \sum_{u\in U_{q+1}} u^{j-i}\\
&= a_i,
\end{array}
\end{eqnarray*}
where the last equality comes from (\ref{eq:power-sum}). The desired conclusion then follows.
\end{proof}

The following lemma gives the first two terms of the polynomial
expansion for the function $\left ( \frac{u+c^q}{cu+1} \right )^e$ over $U_{q+1}$.
\begin{lemma}\label{lem:a1-neq-0}
Let $q=2^m$ with $m \ge 2$,  and $c \in \gf(q^2)^* \setminus U_{q+1}$. Let $e$ be an integer such that $1 \le e \le q$.
Let $a_0 + a_1 u + \cdots + a_q u^q$ be the polynomial expansion of
the function from $U_{q+1}$ to $\gf(q^2)$ given by 
$u \mapsto \left ( \frac{u+c^q}{cu+1} \right )^e$. Then $a_0=0$ and $a_1 = c^{q(e-1)}$.

\end{lemma}

\begin{proof}
Applying Lemma \ref{lem:ai=sum-f} to the function $f(u)=\left ( \frac{u+c^q}{cu+1} \right )^e$, we obtain
\begin{eqnarray*}
\begin{array}{rll}
a_1 &= \sum_{u \in U_{q+1}} \left ( \frac{u+c^q}{cu+1} \right )^e u^{-1} &\\
&= \sum_{u \in U_{q+1}} \left ( \frac{u^{-1}+c^q}{cu^{-1}+1} \right )^e u & \text{ Substituting } u^{-1} \text{ with } u\\
&= \sum_{u \in U_{q+1}} \left ( \frac{1+c^q u}{c+u} \right )^e u & \\
&= \sum_{u \in U_{q+1}} \left ( \frac{u+c}{c^q u+1} \right )^{-e} u & \\
&= \sum_{u \in U_{q+1}} u^{-e} \left ( \frac{u+c}{c^q u+1} \right ) & \text{ Substituting } \frac{u+c}{c^q u+1}  \text{ with } u \\
&= c^{q(e-1)}, &
\end{array}
\end{eqnarray*}
where the last equality follows from Lemmas \ref{lem:frac-poly} and \ref{lem:ai=sum-f}.

Employing Lemma \ref{lem:ai=sum-f} on $\left ( \frac{u+c^q}{cu+1} \right )^e$ again, we have
\begin{eqnarray*}
\begin{array}{rll}
a_0 &= \sum_{u \in U_{q+1}} \left ( \frac{u+c^q}{cu+1} \right )^e  &\\
&= \sum_{u \in U_{q+1}} u^e  & \text{ Substituting } \frac{u+c^q}{cu+1} \text{ with } u\\
&= 0, &
\end{array}
\end{eqnarray*}
where the last equality follows from (\ref{eq:power-sum}).
This completes the proof.

\end{proof}

Now we are ready to prove the main result of this section.

\begin{theorem}\label{thm:code-PGL-4}
Let $q=2^m$ with $m\ge 2$. If $\mathcal C$ is a linear code over $\gf(2^h)$ of length $2^{m}+1$
that is invariant under the permutation action of $\PGL_2(\gf(2^m))$,
then $\mathcal C$ must be one of the following:
\begin{enumerate}
\item[(I)] the zero code $\mathcal C_0 = \{(0,0, \cdots  ,0)\}$; or
\item[(II)] the whole space $\gf(2^h)^{q+1}$, which is the dual of $\mathcal C_0$; or
\item[(III)] the repetition code $\mathcal C_1 = \{(c,c, \cdots  ,c): c \in \gf(2^h)\}$ of dimension $1$; or
\item[(IV)]  the code $\mathcal C_1^{\perp}$,  given by
\[\mathcal C_1^{\perp}= \left \{ (c_0, \cdots, c_{q})\in \gf(2^h)^{q+1}: c_0+ \cdots +c_{q}=0 \right \}.\]
\end{enumerate}

\end{theorem}
\begin{proof}
It is evident that the four trivial $2^h$-ary linear codes $\mathcal C_0, \mathcal C_0^{\perp}, \mathcal C_1$ and
$\mathcal C_1^{\perp}$ of length $2^m+1$ are invariant under $\PGL_2(\gf(2^m))$.

Let $\mathcal C$
be a $2^h$-ary linear code of
 length $q+1$ which is invariant
under $\PGL_2(\gf(q))$, which amounts to saying that $\mathcal C$ is invariant under $\mathrm{Stab}_{U_{q+1}}$ by Proposition \ref{prop:Stab-PGL}.
 By Part (I) of Proposition \ref{prop:three-types}, the translation $\pi(u)=u_0 u$ belongs to $\mathrm{Stab}_{U_{q+1}}$, where $u_0 \in U_{q+1}$.
 This clearly forces $\mathcal C$ to be a cyclic code. Let $E$ be the cyclicity-defining set of $\mathcal C$. We consider the following four cases for $E$.

 If $E= \emptyset $, then $\mathcal C= \mathcal C_0$

 If $E= \{ 0 \}$, then $\mathcal C= \mathcal C_1$

 If $\{0\} \subsetneq E$, then there exists an $e\in E \setminus \{0\}$. Applying Corollary \ref{cor:cyclic-extension},  the lifted  code  $\gf(q^{2h}) \otimes \mathcal C$
 to $\gf(q^{2h})$ is the cyclic code over $\gf(q^{2h})$ with respect to the cyclicity-defining set $E$.
 We see at once that $\gf(q^{2h}) \otimes \mathcal C$ also stays invariant under  $\mathrm{Stab}_{U_{q+1}}$ from the definition of  lifting of a cyclic code.
Combining Corollary \ref{cor:cyclic-extension} with Proposition \ref{prop:three-types} we obtain
 $ \left ( \left (\frac{u + c^{q+1}}{ cu +1} \right )^e \right )_{ u\in U_{q+1}} \in \gf(q^{2h}) \otimes \mathcal C$, where $c \in \gf(q^2)^* \setminus U_{q+1}$.
Applying Corollary \ref{cor:a_e-neq-0-E} and Lemma \ref{lem:a1-neq-0}  we can assert that $1\in E$. Thus
$ \left ( \frac{u + c^{q+1}}{ cu +1} \right )_{ u\in U_{q+1}} \in \gf(q^{2h}) \otimes \mathcal C$.
Combining Corollary \ref{cor:a_e-neq-0-E} and Lemma \ref{lem:frac-poly} we deduce $E=\{0,1, \cdots, q\}$.
 We thus get $\mathcal C= \gf(2^h)^{q+1}=\mathcal C_0^{\perp}$.

If $E \neq \emptyset$ and $ 0 \not \in E$, then there exists an $e\in E \setminus \{0\}$.
An analysis similar to that in the proof of the case of $\{0\} \subsetneq E$ shows
that $E=\{1, \cdots, q\}$ and $\mathcal C= \mathcal C_1^{\perp}$.
 This completes the proof.
\end{proof}

The remainder of this section will be devoted to determining the
$2$-rank of some special incidence structures.
Let $\bD = (X, \cB)$ be an incidence structure.
The points of $X$ are usually indexed with $p_1,p_2,\cdots,p_v$, and the
blocks of $\cB$
are normally denoted by $B_1, B_2, \cdots, B_b$.
 The {\em incidence matrix} $M_\bD=(m_{ij})$ of $\bD$
is a $b \times v$ matrix where $m_{ij}=1$ if  $p_j \in B_i$
and $m_{ij}=0$ otherwise.
The $p$-rank of an incidence structure  $\bD$ is defined as the rank of its incidence matrix
over a finite field of characteristic $p$ and denoted by $\mathrm{rank}_p(\bD)$.
 The binary matrix $M_{\bD}$ can be viewed as a matrix over $\gf(q)$ for any
prime power $q$, and its row vectors span a linear code of length $v$
over $\gf(q)$, which is
denoted by $\mathcal C_q(\bD)$ and called the \emph{code} of $\bD$ over $\gf(q)$.
The $p$-rank of incidence structures, i.e., the dimension of the corresponding codes,
can be used to classify incidence structures of certain type.
For example, the $2$-rank and $3$-rank of Steiner triple and quadruple
 systems were intensively studied and
employed for counting and classifying Steiner triple and quadruple systems
 \cite{JMTW},
 \cite{JT}, \cite{SXK}
 \cite{Z16}, \cite{ZZ12}, \cite{ZZ13}, \cite{ZZ13a}.

For any set $A$ and a positive integer $k$, recall that $\binom{A}{k}$ denotes the set of all $k$-subsets of $A$. 
The following theorem is an important corollary of Theorem 
\ref{thm:code-PGL-4}.

\begin{theorem}\label{thm;2-rank}
Let $\cB \subseteq \binom{\PG(1,2^m)}{k}$ such that $m\ge 2$, $1\le k \le 2^m$ and $\cB$ is invariant under the action of 
$\PGL_2(\gf(2^m))$.
Then the $2$-rank of the incidence structure $\bD=(\PG(1,2^m), \cB)$ is given by
\begin{eqnarray*}
\mathrm{rank}_2(\bD)=\left \{
\begin{array}{cl}
2^m, & \text{ if } k \text{ is even},\\
2^m+1, & \text{ if } k \text{ is odd}.
\end{array}
\right .
\end{eqnarray*}

\end{theorem}

\begin{proof}
Since $\cB$ is invariant under the action of $\PGL_2(\gf(2^m))$, then so is the code $\mathcal C_2(\bD)$ of $\bD$. It then follows  
from Theorem \ref{thm:code-PGL-4} that $\mathcal C_2(\bD)=\mathcal C_1^{\perp}$ or $\mathcal C_2(\bD)=\gf(2)^{2^m+1}$.
The desired conclusion then follows.
\end{proof}

\section{Linear codes of length $2^m+1$ with sets of  supports  invariant under 
$\PGL_2(\gf(2^m))$}\label{sec:supports-PGL}

Throughout this section, let $q=2^m$, and let $U_{q+1}$ be the set of all $(q+1)$-th roots of unity in $\gf(q^2)$, 
where  $m\ge 2$ is a positive integer. 
In this section, we describe two families of nontrivial  linear codes with the set of the supports of all codewords 
of any fixed weight being invariant under $\PGL_2(\gf(q))$.

We define a cyclic code over $\gf(q^2)$ of length $q+1$ by   
\begin{eqnarray}
\CE= \left \{
\begin{array}{c}
\left ( a_3 u^3 + a_{q-2} u^{q-2} +a_5 u^5 + a_{q-4} u^{q-4}  \right )_{u \in U_{q+1}} : \\
a_3,  a_{q-2}, a_5 ,  a_{q-4} \in \gf(q^2)
\end{array}
\right \}.
\end{eqnarray}
We index the coordinates of the codewords in $\CE$ and related codes with the elements in $U_{q+1}$. 
It is evident that the dual of $\CE$ is given as
\begin{eqnarray}
\CE^{\perp}= \left \{
\begin{array}{c}
(c_u)_{u \in U_{q+1}}\in \gf(q^2)^{q+1}: \sum_{u\in U_{q+1}} c_u \mathbf{h}_u =\mathbf{0}
\end{array}
\right \},
\end{eqnarray}
where $\mathbf{h}_u$ is the transpose of the row vector 
$(u^{-5}, u^{-3}, u^3, u^5)$.

 It is obvious that if $\left (c_u \right )_{u\in U_{q+1}} \in \CE $ (resp., $\left (c_u \right )_{u\in U_{q+1}} \in \CE^{\perp} $),
then  $\left (c_u^q \right )_{u\in U_{q+1}} \in \CE $ (resp., $\left (c_u^q \right )_{u\in U_{q+1}} \in \CE^{\perp} $).
From \cite[Lemma 7]{GP10}  we deduce that 
\begin{eqnarray}
\tr_{q^2/q}\left (\CE \right ) =  \left . \CE \right  |_{\gf(q)},
\end{eqnarray}
and
\begin{eqnarray}
\tr_{q^2/q}\left (\CE^{\perp} \right ) =  \left . \CE^{\perp} \right  |_{\gf(q)}.
\end{eqnarray}
In fact, $\C_{\{3,5\}}$ is the lifted code of $\tr_{q^2/q}\left (\CE \right )$ to $\gf(q^2)$ and has cyclicity-defining set $\{3,5,q-2, q-4\}$. 
Similarly, $\C_{\{3,5\}}^\perp$ is the lifted code of $\tr_{q^2/q}\left (\CE^\perp \right )$ to $\gf(q^2)$. The reader is referred to Theorem \ref{thm:cyclic-trance} and Corollary \ref{cor:cyclic-extension} for further clarification. 

In order to describe the  supports of the codewords of
 $\tr_{q^2/q}\left (\CE \right )$
and $\left . \CE^{\perp} \right  |_{\gf(q)}$,
we need to employ symmetric polynomials and elementary symmetric polynomials.
A polynomial $f$ is said to be symmetric if it is invariant under any permutation of its variables.
The \emph{elementary symmetric polynomial} (\emph{ESP}) of degree $\ell$ in 
$k$ variables $u_1, u_2, \cdots, u_k$, written  $\sigma_{k,\ell}$, is defined by
\begin{align}\label{eq:esp}
\sigma_{k,\ell}(u_1, \cdots, u_{k})= \sum_{I\subseteq [k], |I|=\ell}  \prod_{j\in I} u_j,
\end{align}
where $[k]= \{1,2, \cdots, k\}$. Already known to Newton, the fundamental theorem of symmetric polynomials asserts that
any symmetric polynomial is a polynomial in the elementary symmetric polynomials.
For any  $k$-variable symmetric polynomial $f$ with coefficients in $\gf(q^2)$, write
\begin{align}\label{eq:sp-B}
\cB_{f,q+1}=\left \{ \{u_1, \cdots, u_k\} \in \binom{U_{q+1}}{k} : f(u_1, \cdots, u_k)=0 \right \}.
\end{align}
In general, it is difficult to determine $| \cB_{f,q+1}|$. However, it was shown in \cite{TD20} that 
\begin{eqnarray}\label{eq:ESP-5-2}
| \cB_{\sigma_{5,2}, q+1}| = \left\{
\begin{array}{ll}
 \frac{1}{10} \binom{q+1}{3}, &  \text{ if $m$ is even, } \\
0, & \text{ if $m$ is odd. } 
\end{array}
\right.
\end{eqnarray}

To determine the parameters of $\tr_{q^2/q}\left (\CE \right )$
and $\left . \CE^{\perp} \right  |_{\gf(q)}$, we prove several lemmas below.
To simplify notation and expressions below, we use 
$\sigma_{k, \ell}$ to denote $\sigma_{k, \ell}(u_1, \ldots, u_\ell)$ for any 
 $\{u_1, \ldots, u_k\} \in  \binom{U_{q+1}}{k}$ whenever $\{u_1, \ldots, u_k\}$ is specified. 

\begin{lemma}\label{lem:sigmasigma}
Let $\sigma_{3,1}, \sigma_{3,2}, \sigma_{3,3}$ be the ESPs given by (\ref{eq:esp}) with $\{u_1,u_2,u_3\} \in  \binom{U_{q+1}}{3}$. Then
\begin{enumerate}
\item[(I)]  $\sigma_{3,1}\sigma_{3,2}+\sigma_{3,3}=(u_1+u_2)(u_2+u_3)(u_3+u_1)$; 
\item[(II)] $\sigma_{3,1}\sigma_{3,2}+\sigma_{3,3}\neq 0$; and 
\item[(III)] $\sigma_{3,2}^2+\sigma_{3,1}\sigma_{3,3}= \sigma_{3,3}^2\left (\sigma_{3,1}^2+ \sigma_{3,2} \right )^q$. 
\end{enumerate} 
\end{lemma}
\begin{proof}
The proofs are straightforward and omitted.
\end{proof}

\begin{lemma}\label{lem:sigma3not=0}
Let $q=2^m$ with $m$ even. Let $\sigma_{3,1}, \sigma_{3,2}, \sigma_{3,3}$ be the ESPs given by (\ref{eq:esp}) with
$\{u_1,u_2,u_3\} \in  \binom{U_{q+1}}{3}$. Then
\begin{enumerate}
\item[(I)] $\sigma_{3,1}^2+ \sigma_{3,2}\neq 0$; and
\item[(II)] $\sigma_{3,2}^2+\sigma_{3,1}\sigma_{3,3} \neq 0$. 
\end{enumerate} 
\end{lemma}
\begin{proof}
The proof can be found in \cite{TD20}. 
\end{proof}

 For a positive integer $\ell \leq q+1$, define a $4\times  \ell$ matrix $M_{\ell}$ by
\begin{eqnarray}\label{eq:M}
\left [
\begin{array}{llll}
u_1^{-5} & u_2^{-5} & \cdots & u_{\ell}^{-5} \\
u_1^{-3} & u_2^{-3} & \cdots & u_{\ell}^{-3} \\
u_1^{+3} & u_2^{+3} & \cdots & u_{\ell}^{+3} \\
u_1^{+5} & u_2^{+5} & \cdots & u_{\ell}^{+5} \\
\end{array}
\right ],
\end{eqnarray}
where $u_1, \cdots,  u_{\ell} \in U_{q+1}$. For $r_1, \cdots, r_i \in \{ \pm5, \pm3\}$, let $M_{\ell}[r_1, \cdots, r_i]$
denote the submatrix of $M_{\ell}$  obtained by deleting   the rows
$(u_1^{r_1},  u_2^{r_1},  \cdots ,u_{\ell}^{r_1})$, $\cdots$ ,  $(u_1^{r_i} , u_2^{r_i} , \cdots , u_{\ell}^{r_i})$ of the matrix $M_{\ell}$, 
where $1 \leq i \leq 4$. 

\begin{lemma}\label{lem:sol-rank}
Let $M_{\ell}$ be the matrix given by (\ref{eq:M}) with  $\{u_1, \cdots,  u_{\ell}\} \in  \binom{U_{q+1}}{\ell}$.
Consider the system of homogeneous linear equations defined by
\begin{align}\label{eq:Mx=0}
M_{\ell} (x_1, \cdots,  x_{\ell})^{T}=0.
\end{align}
Then (\ref{eq:Mx=0}) has a nonzero  solution $(x_1, \cdots,  x_{\ell})$ in $\gf(q)^{\ell}$  if and only if
$\rank (M_{\ell})<\ell$, where $\rank (M_{\ell})$ denotes the rank of  the matrix $M_{\ell}$.
\end{lemma}

\begin{proof}
The proof is similar to that in \cite[Lemma 29]{TD20} and thus omitted.
\end{proof}

\begin{lemma}\label{lem:rank44}
Let $m$ be an even positive integer and $M_{3}$ be the matrix given by (\ref{eq:M}) with  $\{u_1, u_2, u_3\} \in  \binom{U_{q+1}}{3}$.
Then   $\rank (M_{3})=3$.
\end{lemma}
\begin{proof}
Suppose that $\rank (M_{3})<3$. Then $\det (M_3[5])=\frac{\prod_{1\le i<j\le 3} (u_i+u_j)^2}{\sigma_{3,3}^5}   \left (\sigma_{3,1}^2+\sigma_{3,2} \right )^2$=0,
which is contrary to Lemma \ref{lem:sigma3not=0}. This completes the proof.
\end{proof}

\begin{lemma}\label{lem:rank65}
Let $m$ be an even positive integer and $M_{4}$ be the matrix given by (\ref{eq:M}) with  $\{u_1, \cdots,  u_4\} \in  \binom{U_{q+1}}{4}$.
Then   $\rank (M_{4})=3$ if and only if $\sigma_{4,2}^2+\sigma_{4,1}\sigma_{4,3} =0$.
\end{lemma}

\begin{proof}
Note that
$$\det (M_4)= \frac{\prod_{1\le i<j\le 4} (u_i+u_j)^2}{\sigma_{4,4}^5} \left (\sigma_{4,2}^2+\sigma_{4,1}\sigma_{4,3} \right )^2,$$
which completes the proof.
\end{proof}

The following lemma is immediate from \cite[Lemmas 18 and 20]{TD20}.
\begin{lemma}\label{lem:u^2+au+b=0}
Let $q=2^m$ with $m$ even and $\{u_1, u_2, u_3\} \in \binom{U_{q+1}}{3}$.
Let $a=\sigma_{3,1}^2+ \sigma_{3,2}$, $b=\sigma_{3,1}\sigma_{3,2}+\sigma_{3,3}$
and $c=\sigma_{3,2}^2+\sigma_{3,1}\sigma_{3,3}$. Then
the quadratic polynomial $a u^2 + bu + c$ has exactly two roots $u_4, u_5$ in $U_{q+1}$ such that $\{u_1,u_2,u_3,u_4,u_5\} \in \binom{U_{q+1}}{5}$.
Moreover, $\{u_1, u_2, u_3, u_4\}$ satisfies $\sigma_{4,2}^2+\sigma_{4,1}\sigma_{4,3}=0$.
\end{lemma}

\begin{lemma}\label{lem:wt6-sp}
Let $q=2^m$ with $m$ even and $M_{4}$ be the matrix given by (\ref{eq:M}) with  $\{u_1, \cdots,  u_4\} \in  \binom{U_{q+1}}{4}$.
If there exists a vector $(x_1, \cdots, x_4) \in \left ( \gf(q)^*\right )^4$ such that $M_4(x_1, \cdots, x_4)^T=0$,  then
 $\{u_1, \cdots, u_4\} \in \cB_{\sigma_{4,2}^2+\sigma_{4,1}\sigma_{4,3},q+1}$, where $\cB_{\sigma_{4,2}^2+\sigma_{4,1}\sigma_{4,3},q+1}$ is defined by (\ref{eq:sp-B}).
 \end{lemma}

\begin{proof}
The proof is similar to that in \cite[Lemma 34]{TD20} and thus omitted.

\end{proof}

The minimum-weight codewords in $\tr_{q^2/q}\left (\CE \right )$ are described in the following lemma.

\begin{lemma}\label{lem:trPoly}
Let $f(u)=\tr_{q^2/q} \left ( a u^5+ b u^3 \right )$ where  $(a, b) \in \gf(q^2)^2 \setminus  \{\mathbf 0\}$.
Define
$$\mathrm{zero}(f)=\left \{ u\in U_{q+1} : f(u)=0 \right \}.$$
Then $| \mathrm{zero}(f)  | \le 5$. Moreover, $| \mathrm{zero}(f)| = 5$
if and only if $a=\frac{\tau}{\sigma_{5,5}(u_1, \ldots, u_5)}$ and $b= \frac{\tau\sigma_{5,1}^2(u_1, \ldots, u_5)}{\sigma_{5,5}(u_1, \ldots, u_5)} $,
where $\{u_1, \cdots, u_5\} \in \cB_{\sigma_{5,2},q+1}$ and $\tau \in \gf(q)^*$.
In particular, the dimension of $\tr_{q^2/q}\left (\CE \right )$ equals $4$.
\end{lemma}

\begin{proof}
When $u\in U_{q+1}$, one has
\begin{align}\label{eq:f-c-U}
f(u)=\frac{1}{u^5}\left ( \sqrt{a} u^5 + \sqrt{b} u^4 +\sqrt{b^q}u + \sqrt{a^q}  \right )^2.
\end{align}
Thus $| \mathrm{zero}(f) | \le 5$.

Assume that $| \mathrm{zero}(f) |= 5$. From  (\ref{eq:f-c-U}),
there exists $\{u_1, \cdots, u_5\} \in \binom{U_{q+1}}{5}$
such that $f(u)=\frac{a \prod_{i=1}^5 (u+u_i)^2}{u^5}$.
By Vieta's formula, $a\sigma_{5,1}^2=b$, $a\sigma_{5,2}^2=0$, $a\sigma_{5,3}^2=0$, $a\sigma_{5,4}^2=b^q$ and $a\sigma_{5,5}^2=a^q$.
One obtains $a= \frac{\tau}{\sigma_{5,5}}$ from $a^{q-1}= \sigma_{5,5}^2$, where $\tau \in \gf(q)^*$.
 Thus $b=\frac{\tau\sigma_{5,1}^2}{\sigma_{5,5}}$.

Conversely, assume that $a= \frac{\tau}{\sigma_{5,5}}$ and $b=\frac{\tau\sigma_{5,1}^2}{\sigma_{5,5}} $,
where $\{u_1, \cdots, u_5\} \in \cB_{\sigma_{5,2},q+1}$ and $\tau \in \gf(q)^*$. Then
$f(u)=\frac{a \prod_{i=1}^5(u+u_i)^2}{u^5}$. Consequently, $\mathrm{zero}(f)=\{u_1,  \cdots, u_5\}$ and $|\mathrm{zero}(f)|=5$.
This completes the proof.
\end{proof}

\begin{theorem}\label{thm:dual-C-3-5}
Let $q=2^m$ with $m\ge 4$ being an even integer. Then the subfield  subcode $\left . \CE^{\perp} \right  |_{\gf(q)}$
has parameters $[q+1, q-3, 4]_q$. 
\end{theorem}
\begin{proof}
Recall that (\ref{eq:subfield-trace}) says that
\begin{eqnarray*}
\left . \CE^{\perp}\right |_{\gf(q)}=\left (  \tr_{q^2/q} (\CE) \right )^{\perp}.
\end{eqnarray*}
Thus $\left . \CE^{\perp}\right |_{\gf(q)}$ has dimension $q-3$ by Lemma \ref{lem:trPoly}.

Since $m$ is even, we have $d\ge 4$ by Lemma \ref{lem:rank44}.
Applying Lemmas \ref{lem:u^2+au+b=0} and \ref{lem:wt6-sp}, we assert that
there must exist a codeword of weight $4$. Consequently, $d=4$. 
\end{proof}

As Theorem \ref{thm:dual-C-3-5} showed, 
the subfield  subcode $\left . \CE^{\perp} \right  |_{\gf(q)}$ almost meets the Griesmer bound.

\begin{theorem}\label{thm:C-3-5}
Let $q=2^m$ with $m\ge 4$ being even.
Then the trace code $\tr_{q^2/q}\left (\CE \right )$  has parameters $[q+1, 4, q-4]_q$. 
\end{theorem}

\begin{proof}
Note that any codeword of $\tr_{q^2/q} \left ( \CE \right )$  can be written as
 \[ \mathbf{c}(a_3,a_5)= \tr_{q^2/q}(a_3 u^3+a_5 u^5). \]
Then the dimension of $\tr_{q^2/q}\left (\CE \right )$ is equal to $4$  by Lemma \ref{lem:trPoly}.
The desired conclusion on the minimum weight for even $m$ then follows from (\ref{eq:ESP-5-2}) and Lemma \ref{lem:trPoly}. 
This completes the proof.
\end{proof}

The invariance of the set of the supports of all the codewords of any fixed weight in $\tr_{q^2/q} \left ( \CE \right )$ 
under the action of $\PGL_2(\gf(2^m))$ 
is established by the following theorem.

\begin{theorem}\label{thm:trace-supp-inv}
Let $q=2^m$ with $m\ge 2$. Let $k$ be an integer with $1 \le k \le q+1$ and
$A_k \left (\tr_{q^2/q} \left ( \CE \right )  \right )  >0$.
Then $\mathcal B_k \left (\tr_{q^2/q} \left ( \CE \right )  \right )$ is invariant under the action of 
$\mathrm{Stab}_{U_{q+1}}$.
In particular, the incidence structure $\left ( U_{q+1},  \mathcal B_k \left (\tr_{q^2/q} \left ( \CE \right )  \right )\right )$
is a $3$-design when $k>3$.
\end{theorem}

\begin{proof}
We only need to show that if $\mathbf{c} \in \tr_{q^2/q} \left ( \CE \right )$
and $\pi $ is a linear fractional transformation listed in Corollary \ref{cor:generators-three-types},
then there exists a codeword $\mathbf{c'} \in \tr_{q^2/q} \left ( \CE \right )$
such that $\mathrm{Supp}(\pi (\mathbf{c}))=\mathrm{Supp}(\mathbf{c'})$.
Denote by $\mathbf{c}(a_3,a_5)$ the codeword $\left ( \tr_{q^2/q} (a_3 u^3+ a_5 u^5) \right )_{u \in U_{q+1}}$ of $\tr_{q^2/q} \left ( \CE \right )$,
where $a_3, a_5 \in \gf(q^2)$.  We investigate the following three cases for $\pi$.

If $\pi$ is the transformation given by $u \mapsto u_0 u$, where $u_0 \in U_{q+1}$, then
it is clear that $\pi ( \mathbf{c}(a_3,a_5) )=  \mathbf{c}(a_3 u_0^3,a_5 u_0^5) $.
Thus $\mathrm{Supp} \left ( \pi ( \mathbf{c}(a_3,a_5) ) \right )=  \mathrm{Supp} \left ( \mathbf{c}(a_3 u_0^3,a_5 u_0^5)  \right )$.

If $\pi$ is the transformation given by $u \mapsto  u^{-1}$,  then
it is obvious that $\pi ( \mathbf{c}(a_3,a_5) )=  \mathbf{c}(a_3 ,a_5 ) $.
Thus $\mathrm{Supp} \left ( \pi ( \mathbf{c}(a_3,a_5) ) \right )=  \mathrm{Supp} \left ( \mathbf{c}(a_3 ,a_5 )  \right )$.

Let $\pi$ be the translation given by $u \mapsto \frac{u+c^q}{cu+1}$ where $c \in \gf(q^2)^* \setminus U_{q+1}$.
Write $f(u)= \tr_{q^2/q} (a_3 u^3+ a_5 u^5)$ and $A=cu+1$. Then $u+c^q=uA^q.$
A standard computation gives
\begin{eqnarray}\label{eq:anction-f}
\begin{array}{rl}
\lefteqn{ f\left (\frac{u+c^q}{cu+1} \right ) } \\
&=\tr_{q^2/q}\left ( a_3 \left (\frac{u+c^q}{cu+1} \right )^3 + a_5 \left (\frac{u+c^q}{cu+1} \right )^5 \right )\\
&= \tr_{q^2/q}\left ( \frac{ a_3 (u+c^q)^3(cu+1)^2 + a_5 (u+c^q)^5}{(cu+1)^5} \right )\\
&= \tr_{q^2/q}\left ( \frac{ a_3 A^{3q}A^2u^3 + a_5 A^{5q} u^5}{A^5} \right )\\
&= \frac{ a_3 A^{3q}A^2u^3 + a_5 A^{5q} u^5}{A^5} + \frac{ a_3^q A^{3}A^{2q}u^{3q} + a_5^q A^{5} u^{5q}}{A^{5q}}\\
&= \frac{a_3 A^{8q}A^2u^3 + a_5 A^{10q} u^5 +  a_3^q A^{8}A^{2q}u^{3q} + a_5^q A^{10} u^{5q}}{A^5 A^{5q}}\\
&= \frac{a_3 A^{8q}A^2u^3 + a_5 A^{10q} u^5 +  \left ( a_3 A^{8q}A^2u^3 + a_5 A^{10q} u^5 \right )^q}{A^5 A^{5q}}\\
&= \frac{1}{A^5 A^{5q}} \tr_{q^2/q} \left ( a_3 A^{8q}A^2u^3 + a_5 A^{10q} u^5 \right ).
\end{array}
\end{eqnarray}
Expanding $a_3 A^{8q}A^2u^3$ yields
\begin{equation}\label{eq:expanding-2}
\begin{array}{rl}
\lefteqn{ a_3 A^{8q}A^2u^3} \\
&= a_3 (c^{8q} u^{8q}+1)(c^2u^2+1)u^3\\
&= a_3 u^3 (c^{8q+2} u^{8q+2}+ c^{8q} u^{8q}+ c^2u^2+1 )\\
&= a_3 (u^3+ c^{8q+2} u^{-3} + c^2u^{5}  + c^{8q} u^{-5}).
\end{array}
\end{equation}
Expanding $a_5 A^{10q} u^5$ yields
\begin{equation}\label{eq:expanding-1}
\begin{array}{rl}
\lefteqn{ a_5 A^{10q} u^5} \\
&= a_5 (c^{10q} u^{10q}+1)u^5\\
&= a_5 ( u^{5}  + c^{10q} u^{-5}).
\end{array}
\end{equation}
Combining (\ref{eq:expanding-2}) and (\ref{eq:expanding-1}) gives
\begin{eqnarray}\label{eq:tr-tr}
\begin{array}{rl}
\lefteqn{ \tr_{q^2/q} \left ( a_3 A^{8q}A^2u^3 + a_5 A^{10q} u^5 \right )} \\
&= \tr_{q^2/q} \left ( \left ( a_3+a_3^q c^{8+2q}\right ) u^3 + \left(a_5  + a_5^q c^{10}+ a_3c^2+ a_3^q c_8\right )u^5 \right ).
\end{array}
\end{eqnarray}
Plugging (\ref{eq:tr-tr}) into (\ref{eq:anction-f}) yields
\[f\left (\frac{u+c^q}{cu+1} \right )= \frac{1}{A^5 A^{5q}} \tr_{q^2/q} \left ( a_3' u^3 + a_5' u^5 \right ),\]
where $a_3'=a_3+a_3^q c^{8+2q}$ and $a_5'=a_5  + a_5^q c^{10}+ a_3c^2+ a_3^q c_8$. This clearly forces
$\mathrm{Supp} \left ( \pi (\mathbf{c}(a_3,a_5) ) \right )= \mathrm{Supp} \left ( \mathbf{c}(a_3',a_5')  \right )$.
The desired conclusion then follows.

\end{proof}

The proof of Theorem \ref{thm:trace-supp-inv} gives more, namely
\begin{eqnarray}\label{eq:functions-modular}
\begin{array}{rl}
 \lefteqn{ \tr_{q^2/q} \left ( a_3 \left ( \frac{u+c^q}{cu +1} \right )^3 +a_5 \left ( \frac{u+c^q}{cu +1} \right )^5 \right ) } \\
 &= \frac{1}{(cu+1)^5 (cu+1)^{5q}} \tr_{q^2/q} \left ( a_3' u^3 + a_5' u^5 \right ),
 \end{array}
\end{eqnarray}
 where $a_3, a_5 \in \gf(q^2)$, $c \in \gf(q^2) \setminus U_{q+1}$,
 $a_3'=a_3+a_3^q c^{8+2q}$ and $a_5'=a_5  + a_5^q c^{10}+ a_3c^2+ a_3^q c_8$.

The following theorem shows the invariance of the set of the supports of 
all the codewords of any fixed weight in $ \left . \CE^{\perp}\right |_{\gf(q)}$ under the action of $\PGL_2(\gf(2^m))$. 

\begin{theorem}\label{thm:sub-supp-inv}
Let $q=2^m$ with $m\ge 2$. Let $k$ be any integer with $1 \le k \le q+1$ and
$A_k \left (  \left . \CE^{\perp}\right |_{\gf(q)}\right )  >0$.
Then $\mathcal B_k \left ( \left . \CE^{\perp}\right |_{\gf(q)} \right )$ is invariant under the action of $\mathrm{Stab}_{U_{q+1}}$.
In particular, the incidence structure $\left ( U_{q+1},  \left . \CE^{\perp}\right |_{\gf(q)} \right )$
is a $3$-design when $k>3$.
\end{theorem}

\begin{proof}
Recall that  by (\ref{eq:subfield-trace}) we have
\begin{eqnarray*}
\left . \CE^{\perp}\right |_{\gf(q)}=\left (  \tr_{q^2/q} (\CE) \right )^{\perp}.
\end{eqnarray*}
Let $\mathbf{w}$ be any codeword of $ \left . \CE^{\perp}\right |_{\gf(q)}= \left (  \tr_{q^2/q} (\CE) \right )^{\perp}$
and $\pi$ be any linear fractional translations listed in Corollary \ref{cor:generators-three-types}.
It is easily seen that if $\pi$ is a transformation given by $u \mapsto u_0 u$ or $u \mapsto 1/u$, where $u_0 \in U_{q+1}$, then
\begin{eqnarray}\label{eq:dual-1-2}
\pi (\mathbf{w} ) \in \left . \CE^{\perp}\right |_{\gf(q)}.
\end{eqnarray}
Assume $\pi$ is a translation given by $u \mapsto  \frac{u+c^q}{cu+1}$ where $c \in \gf(q^2)^* \setminus U_{q+1}$.
It is obvious that  $\pi (\mathbf{w}) \in \left ( \pi\left ( \tr_{q^2/q} (\CE)  \right )  \right )^{\perp}$.
From (\ref{eq:functions-modular}) we conclude that 
\[ \pi\left ( \tr_{q^2/q} (\CE)  \right ) =\left  (\frac{1}{(cu+1)^{5q+5}} \right )_{u \in U_{q+1}} \cdot  \tr_{q^2/q} (\CE).\]
By (\ref{eq:a-code-dual}) we have that 
$$ \left ( \pi\left ( \tr_{q^2/q} (\CE)  \right )  \right )^{\perp}=
\left  ((cu+1)^{5q+5} \right )_{u \in U_{q+1}} \cdot  \left ( \tr_{q^2/q} (\CE) \right )^{\perp}.
$$ 
Consequently,
\begin{eqnarray}\label{eq:dual-3}
\pi (\mathbf{w}) \in \left  ((cu+1)^{5q+5} \right )_{u \in U_{q+1}} \cdot  \left ( \tr_{q^2/q} (\CE) \right )^{\perp}.
\end{eqnarray}
Combining (\ref{eq:dual-1-2}) and (\ref{eq:dual-3}) with Corollary \ref{cor:generators-three-types}
we can assert that the set of all the supports of $ \left . \CE^{\perp}\right |_{\gf(q)}$
stays invariant under $\mathrm{Stab}_{U_{q+1}}$.
This completes the proof. 
\end{proof}

The remainder of this section is devoted to determining the parameters of certain $3$-designs held 
in the subfield subcodes $\left . \CE^{\perp}\right |_{\gf(q)}$ and the trace codes
$  \tr_{q^2/q} (\CE)  $.

\begin{theorem}\label{thm:10-5}
Let $q=2^m$ with $m \geq 4$ even. Then the incidence structure 
$$\left ( U_{q+1}, \cB_{q-4} \left (\tr_{q^2/q} \left ( \CE \right )   \right )  \right )$$
is a $3$-$(q+1, q-4, \lambda)$ design with 
$$ 
\lambda=\frac{(q-4)(q-5)(q-6)}{60}, 
$$ 
and its complementary incidence structure  is a $3$-$(q+1,5,1)$ design.
\end{theorem}

\begin{proof} 
By Theorems \ref{thm:trace-supp-inv} and \ref{thm:C-3-5},  $\left ( U_{q+1}, \cB_{q-4} \left (\tr_{q^2/q} \left ( \CE \right )   \right )  \right )$ is a $3$-$(q+1, q-4, \lambda)$ design. To determine the value of $\lambda$, we consider its complementary design.  
Lemma \ref{lem:trPoly} shows that 
the complementary  incidence structure of   $$\left ( U_{q+1}, \cB_{q-4} \left (\tr_{q^2/q} \left ( \CE \right )   \right )  \right )$$ is isomorphic to
the  incidence structure of $(U_{q+1}, \cB_{\sigma_{5,2}, q+1})$, which is a $3$-$(2^m+1,5,1)$ design by \cite[Theorem 5]{TD20}.
It the follows that 
$$ 
\lambda = \frac{\binom{q+1-3}{5}}{\binom{q+1-3}{5-3}}=
\frac{(q-4)(q-5)(q-6)}{60}. 
$$
This completes the proof.
\end{proof}
The complementary design $\bar{\bD}$
of the 3-design $\bD$ from Theorem \ref{thm:10-5}
has parameters 3-$(2^m +1,5,1)$. 
These parameters correspond to a spherical geometry design \cite[Volume I, page 193]{BJL}.
If a linear code $\C$ supports a $t$-design $\bD$, it is in general an open
question  how to construct a linear code $\C'$  that supports the complementary design of 
$\bD$. 
We note that an isomorphic version of the complementary $3$-$(2^m+1,5,1)$ design 
$\bar{\bD}$
of the design $\bD$ from Theorem \ref{thm:10-5}   is supported by a linear code 
described in \cite{TD20}. 
According to Magma experiments, $\bar{\bD}$ is isomorphic to a spherical 
geometry design with the same parameters when $m \in \{4, 6\}$.
 The following theorem asserts that the $3$-$(2^m+1,5,1)$ design $\bar{\bD}$ is isomorphic to the spherical geometry 
 design\footnote{See \cite[Volume I, 6.9 and 6.10, page 193]{BJL} for a short description of the spherical geometry designs found by Witt \cite{W}.}
  found by Witt  \cite{W} in general. 

\begin{theorem}
\label{t26}
Let $q=2^m$ with $m \geq 4$ even. Then the complementary design $\bar{\bD}$
of the $3$-design $\bD$ from Theorem \ref{thm:10-5} is isomorphic to the 
Witt spherical geometry design \cite{W} with parameters $3$-$(2^m+1,5,1)$.
\end{theorem}
\begin{proof}
By definition, $\bar{\bD}=\left ( U_{q+1}, \bar{\cB} \right )$,
where $\bar{\cB}= \left \{ B\in \binom{U_{q+1}}{5}: \left (  U_{q+1} \setminus B \right )  \in  \cB_{q-4} \left (\tr_{q^2/q} \left ( \CE \right )   \right ) \right \}$.
Theorem \ref{thm:trace-supp-inv} now implies that $\bar{\cB}$ is invariant under  $\mathrm{Stab}_{U_{q+1}}$. 
Applying Proposition \ref{prop:Stab-PGL} we conclude that $\bar{\bD}$ is isomorphic to 
a $3$-$(2^m+1,5,1)$ design $\left ( \PG(1,2^m)  , \cB \right )$ with $\cB$ being invariant under $\PGL_2(\gf(2^m))$.
Hence we can write $\cB$ as $\cB=\dot\cup_{i=1}^{\ell} \mathrm{Orb}_{B_i}$, where $B_i\in \binom{\PG(1,2^m)}{5}$
and $\mathrm{Orb}_{B_i}$ is the $\PGL_2(\gf(2^m))$-orbit of $B_i$. 
Let $\mathrm{Stab}_{B_i}$ denote 
the stabilizer of $B_i$ under the action of $\PGL_2(\gf(2^m))$ on $\binom{\PG(1,2^m)}{5}$. Let $B$ be any $5$-subset of $\PG(1,2^m)$.
Let us recall that $\left | \mathrm{Stab}_{B} \right | \in \{1,4,60\}$ and all $5$-subsets $B$ of $\PG(1,2^m)$ with $\left | \mathrm{Stab}_{B} \right |=60$
form exactly one $\PGL_2(\gf(2^m))$-orbit
 (see {Huber05}). 
 A trivial verification shows that the $5$-subset $\PG(1,4)=\gf(4) \cup  \left \{ \infty  \right \} $ of $\PG(1,2^m)$  is stabilized by
$\PGL_2(\gf(4))$. As the cardinality of the group $\PGL_2(\gf(4))$ is $60$ we have $\left | \mathrm{Orb}_{\PG(1,4)} \right |
=\left | \PGL_2(\gf(2^m)) \right | /60$ and $\left \{ B\in \binom{\PG(1,2^m)}{5} : \left | \mathrm{Stab}_{B} \right |=60 \right \}=\mathrm{Orb}_{\PG(1,4)}$.
Let us observe that $\left | \mathrm{Orb}_{B_i} \right |=\left | \PGL_2(\gf(2^m)) \right | /60, \left | \PGL_2(\gf(2^m)) \right | /4, \mbox{or} 
\left | \PGL_2(\gf(2^m)) \right |$ and 
there is exactly one orbit $\mathrm{Orb}_{\PG(1,4)}$ with size $=\left | \PGL_2(\gf(2^m)) \right | /60$.
It follows that $\ell =1$ and $\cB=\mathrm{Orb}_{\PG(1,4)}$ from $\left | \cB \right |= \left | \PGL_2(\gf(2^m)) \right | /60$.
The desired conclusion then follows from the definition of spherical geometry designs (see for instance \cite[Volume I, page 193]{BJL}).

\end{proof}

\begin{theorem}\label{thm:code-design-esp-odd}
Let $q=2^m$ with $m \geq 4$ being even. Then, the incidence structure $$\left (U_{q+1}, \cB_4\left ( \left . \CE^{\perp}\right |_{\gf(q)}  \right )  \right )$$
supported by
 the minimum-weight codewords in $ \left . \CE^{\perp}\right |_{\gf(q)} $ is a $3$-$(q+1,4,2)$ design.
\end{theorem}

\begin{proof} 
By Theorems \ref{thm:sub-supp-inv} and \ref{thm:dual-C-3-5},   $\left (U_{q+1}, \cB_4\left ( \left . \CE^{\perp}\right |_{\gf(q)}  \right )  \right )$ is a $3$-$(q+1, 4, \lambda)$ design. It remains to determine the value of $\lambda$. But combining Lemmas \ref{lem:wt6-sp} and \ref{lem:u^2+au+b=0} yields directly that it is a  $3$-$(q+1, 4, 2)$ design. 
\end{proof}

It would be interesting to determine parameters for more $3$-designs held in $\tr_{q^2/q} \left ( \CE \right ) $ and $ \left . \CE^{\perp}\right |_{\gf(q)}$. 
To the best knowledge of the authors, Theorem \ref{thm:code-design-esp-odd} documents the first infinite family of 
linear codes supporting an infinite family of $3$-$(v, 4, 2)$ designs. According to \cite[Table 4.37, page 83]{CRC},  a class
of $3$-$(q+1,4,2)$  designs with $q\equiv 1 \pmod 3$ were found by
Hughes \cite{Hu}. We checked with Magma
\cite{magma}
 that in the cases $m=4$ and $m=6$,
the 3-$(17,4,2)$ design from Theorem \ref{thm:code-design-esp-odd} 
 and the design with these parameters
found in \cite{Hu} are isomorphic. In case the two $3$-$(q+1,4,2)$ designs are isomorphic 
for every even $m \geq 4$, 
the contribution of Theorem \ref{thm:code-design-esp-odd} will be a coding-theoretic 
construction of the $3$-$(q+1,4,2)$ designs.

\begin{example} 
Let $q=2^4$. Then $\tr_{q^2/q}(\C_{\{3,5\}})$ has parameters $[17,4,12]_{16}$ and weight enumerator 
$$ 
1 + 1020z^{12} + 24480z^{15} + 15555z^{16} + 24480z^{17},  
$$ 
and $$\left ( U_{q+1}, \cB_{q-4} \left (\tr_{q^2/q} \left ( \CE \right )   \right )  \right )$$
is a $3$-$(17, 12, 22)$ design. 

The code  $ \left . \CE^{\perp}\right |_{\gf(q)} $ has parameters $[17,13,4]_{16}$ and weight enumerator 
\begin{eqnarray*}
1 + 5100z^4 + 42840z^5 + 2244000z^6 +  50669520z^7 +  949969350z^8 + \\ 
14262976200z^9 + 171117027840z^{10} + 1633451574240z^{11} + \\
12250821846060z^{12} + 70677865367400z^{13} + 302905113919200z^{14} + \\ 
 908715349415760z^{15} + 1703841278658465z^{16} + 1503389363654520z^{17}, 
\end{eqnarray*} 
and $\left (U_{q+1}, \cB_4\left ( \left . \CE^{\perp}\right |_{\gf(q)}  \right )  \right )$ is a $3$-$(17, 4,2)$ design. 
\end{example} 

\section{On the $q$-dimension of $3$-$(q+1,q-4,(q-4)(q-5)(q-6)/60)$
and $3$-$(q+1,4,2)$ designs}
\label{Sec6}

In this section, we discuss the $q$-dimension of the 3-designs documented in Section 
\ref{sec:supports-PGL}. 
Recall the $q$-dimension of $t$-designs introduced in \cite{Tdim} and the introduction
of this paper.
An obvious upper bound on the $q$-dimension is the dimension of the supporting code.
We will use the following lemma to derive a lower bound.

\begin{lemma}
\label{fx}
Let $f(x)=x^3 -60x^2 -61x -60$. Then $f(x)>0$ for every $x\ge 62$, and $f(x)<0$ for 
$0 \le x\le 61$.
\end{lemma}
\begin{proof}
The derivative $f'(x)=3x^2 -120x-61$ has roots $20\pm \sqrt{3783}/3$.
It follows that $f(x)$ is decreasing on the interval $(20-\sqrt{3783}/3, 20+\sqrt{3783}/3)$,
and increasing on the interval $(20+\sqrt{3783}/3, \infty)$. Note that 
$\sqrt{3783}/3 \approx 20.5$. Hence, $f(x)$ is decreasing on the interval $(0, 40)$,
and increasing on the interval $(41, \infty)$.  Since
\[ f(0) = f(61)=-60, \ f(62)=3846, \]
the lemma follows.
\end{proof}

\begin{theorem}
\label{bound}
Suppose that $\bD$ is a 3-$(q+1,q-4,(q-4)(q-5)(q-6)/60)$ design, where
$q$ is a prime power. If $q>63$ then
\[ \dim_{q}\bD \ge 4, \]
where $\dim_{q}\bD$ is the dimension of $\bD$ over the finite field $\gf(q)$ of order $q$.
\end{theorem}

\begin{proof} The number of blocks of $\bD$ is
\begin{equation}
\label{bb}
 b=\frac{(q+1)q(q-1)}{60} =\frac{q^3 -q}{60}. 
 \end{equation}

If $\C$ is a linear code over $\gf(q)$ of length $q+1$, such that every block of $\bD$ is
the support of a codeword of weight $q-4$, $\C$ must contain at least $b(q-1)$
codewords of weight $q-4$.  It is sufficient to show that
\begin{equation}
\label{bq}
b(q-1) > q^3 -1, 
\end{equation}
which would imply that $|\C| \ge q^4$, hence, the dimension of $\C$ is greater that or equal to 4.
Substituting $b$ in (\ref{bq}) by the right-hand side of eq. (\ref{bb}) implies that
the inequality (\ref{bq}) is equivalent to
\begin{equation}
\label{le}
  q^3 -60q^2 -61q -60 > 0. 
 \end{equation}
  Since $q>63$, the  inequality (\ref{le}) holds by Lemma \ref{fx}.
 \end{proof}
  As a corollary of Theorem \ref{bound}, we have the following.
  \begin{theorem}
  \label{qdim}
  The 3-$(q+1,q-4,(q-4)(q-5)(q-6)/60)$ design $\bD$ from Theorem \ref{thm:10-5}
has dimension 4 over $\gf(2^m)$ for every even $m\ge 6$.
  \end{theorem}
  \begin{proof}
  The blocks of the design $\bD$ are supprts  of minimum weight codwords in the
  $[q+1, 4, q-4]_q$ code $\C$ with $q=2^m$, $m\ge 4$ even, from Theorem \ref{thm:C-3-5}.
Since the dimension of $\C$ is 4, it follows that $\dim_{q}\bD \le 4$.
On the other hand, according to Theorem \ref{bound},
 $\dim_{q}\bD \ge 4$ for $q=2^m \ge 64$, that is, for every even 
$m\ge 6$.
\end{proof}

In the smallest case, $m=4$,  the 3-$(17,12,22)$ design $\bD$
supported by the
$[17,4,12]_{2^4}$ code $\C$ from Theorem  \ref{thm:C-3-5}
 does not
satisfy the hypothesis of Theorem \ref{bound}, thus, we only have
$\dim_{16}\bD \le \dim{\C}=4$.

It turns out that the subfield subcode
$\C' = \C|_{\gf(4)}$ of the $[17,4,12]_{2^4}$ code $\C$ is a $[17,4,12]_4$
code with weight distribution
\[  A_0 = 1, \ A_{12} = 204, \  A_{16} = 51. \]
The 68 distinct supports of codewords of weight 12 in $\C'$ are the blocks
of a 3-$(17,12,22)$ design $\bD'$ identical with the design $\bD$  supported by $\C$.
Since
\[ 3\cdot 68 > 3\cdot 4^3, \]
it follows that
\[ \dim_{4}\bD' = 4. \]

A lower bound 4 on the $q$-dimension of a 3-$(q+1,4,2)$ design for any prime power $q>26$
can be proved as in Theorem \ref{bound}.  However, this bound is far below the upper bound
provided by the dimension $q-3$ of the supporting code from Theorem \ref{thm:dual-C-3-5}.
 The following analysis of the
   $3$-$(17,4,2)$ design, the smallest design in the infinite family of 
  3-designs from Theorem \ref{thm:code-design-esp-odd}, suggests that
  the $q$-dimension is likely to be equal to
  the dimension of the supporting code.
 
The $[17,13,4]_{16}$ code  $\C$ from Theorem \ref{thm:dual-C-3-5} is a cyclic code
with generator polynomial $x^4 + x^3 + \beta^{10}x^2 + x + 1$, where $\beta$
is a primitive element of $\gf(16)$. The following vector is a codeword of weight 4:
\[ u=(1, 0, \beta^5, \beta^5, 0, 1, 0, 0, 0, 0, 0, 0, 0, 0, 0, 0, 0). \]
The twelve cyclic shifts of $u$ form a $12 \times 17$ matrix $M$
of rank 12 in echelon form. Clearly,  replacing  the nonzero entries of $M$  
by arbitrary nonzero elements of $\gf(16)$ changes $M$ to another matrix of rank 12.
It follows that the rank of every generalized $\gf(16)$-incidence matrix of the
$3$-$(17,4,2)$ design $\bD$ from Theorem \ref{thm:code-design-esp-odd} is greater than
or equal to 12. Thus, we have the following.
\begin{theorem}
\label{1742} 
Let $\bD$ be the $3$-$(17, 4, 2)$ design before. Then 
\[ 12 \le \dim_{16}\bD \le 13. \]
\end{theorem}

\section{Summary and concluding remarks}

The main contributions of this paper are the following:
\begin{itemize}
\item A complete classification of linear codes over $\gf(2^h)$ of length $2^m+1$ that are
invariant under the action of the projective general group $\PGL_2(\gf(2^m))$ is
established 
in Theorem \ref{thm:code-PGL-4}.
\item The $2$-ranks of $3$-$(2^m+1,k, \lambda)$ designs that are invariant under the action of $\PGL_2(\gf(2^m))$ are determined
in Theorem \ref{thm;2-rank}.
\item A family of trace codes and a family of subfield subcodes, such that the set of the supports of all codewords of any fixed weight being invariant
under $\PGL_2(\gf(2^m))$, are constructed  in Theorems \ref{thm:trace-supp-inv} and \ref{thm:sub-supp-inv}. 
\item The parameters of the  $3$-designs  supported by the codewords
of minimum weight in these linear codes are presented 
in Theorem \ref{thm:10-5} and Theorem \ref{thm:code-design-esp-odd}.

\item It is proved in Theorem \ref{t26} that 
 the complementary design $\bar{\bD}$
of the $3$-design $\bD$ from Theorem \ref{thm:10-5} is isomorphic to the 
Witt spherical geometry with parameters $3$-$(2^m+1,5,1)$  \cite{W}.

\item A lower bound on the $q$-dimension of
of 3-designs with parameters 3-$(q+1,(q-4),(q-4)(q-5)(q-6)/60)$, 
$q>63$,  is derived in Theorem \ref{bound}, and it  is shown
that  an infinite family of 3-designs  described in Theorem \ref{thm:10-5} meet this bound.

\end{itemize}

We remark that the methodology of this paper may be extended to codes of length $p^m+1$ over $\gf(p)$, where $p$ is 
an odd prime. New linear codes supporting new $t$-designs may be found.




\begin{thebibliography}{99}

\bibitem{AM69} E. F. Assmus Jr., H. F. Mattson Jr., New 5-designs, J. Comb. Theory
6, 122--151, 1969.


\bibitem{BJL}
T. Beth, D. Jungnickel, H. Lenz, Design Theory, Cambridge University
Press, Cambridge, 1999.

\bibitem{magma}
W. Bosma, J. Cannon, Handbook of Magma Functions, School of Mathematics 
and Statistics, University of Sydney, Sydney, 1999.


\bibitem{CRC} C. J. Colbourn, J. F. Dinitz, Handbook of Combinatorial Designs,
Second Edition, Chapman \& Hall/CRC, Boca Raton, 2007.

\bibitem{Dickson01} 
L. E. Dickson. Linear groups: with an exposition of the Galois field theory, Teubner, Leipzig, 1901.

\bibitem{Del75} 
P. Delsarte, On subfield subcodes of modified Reed-Solomon codes, IEEE Trans. Information Theory  21(5), 575--576, 1975.

\bibitem{Dingbook18}
C. Ding, Designs from Linear Codes, World Scientific, Singapore, 2018.

\bibitem{DingTang19}
C. Ding, C. Tang, Infinite families of near MDS codes holding $t$-designs, IEEE Trans. 
Information Theory  66(9),  5419--5428, 2020. 

\bibitem{DWF} 
X. Du, R. Wang, C. Fan, Infinite families of $2$-designs from a class of cyclic codes,  Journal of Combinatorial Designs 28(3), 157--170, 2020.

\bibitem{GP10} 
M. Giorgetti, A. Previtali,  Galois invariance, trace codes and subfield subcodes, 
Finite Fields and Their Applications 16(2), 96--99, 2010. 

\bibitem{Huber05} 
M. Huber, The classification of flag-transitive Steiner 3-designs, 
Advances in Geometry 5(2), 195--221, 2005.  

\bibitem{HP03}
W. C. Huffman, V. Pless, Fundamentals of Error-Correcting Codes,
Cambridge University Press, Cambridge, 2003.

\bibitem{Hu} 
D. R. Hughes, On $t$-designs and groups,
American J. Math. 87(4), 761--778, 1965.

\bibitem{JMTW}
D. Jungnickel, S. S. Magliveras, V. D. Tonchev, A. Wassermann,
The classification of Steiner triple systems on 27 points with 3-rank 24, 
Designs, Codes, and Cryptography  87, 831--839, 2019.

\bibitem{JT13} 
D. Jungnickel and V. D. Tonchev, New invariants for incidence structures,
Designs, Codes and Cryptography 68, 163--177, 2013.

\bibitem{JT}
D. Jungnickel, V. D. Tonchev,
Counting Steiner triple systems with classical parameters and prescribed rank, 
J. Combin. Theory Ser. A  162, 10--33, 2019.

\bibitem{Passman68} 
D. S. Passman, Permutation groups, Benjamin, New York, 1968.

\bibitem{SXK} 
M. Shi, L. Xu, D. S. Krotov, The number of the non-full-rank Steiner triple systems,
J. Comb. Des. 27(10), 571--585, 2019. 


\bibitem{Tang}
C. Tang, Infinite families of $3$-designs from APN functions, 
Journal of Combinatorial Designs 28(2), 97--117, 2020.


\bibitem{TD20}
C. Tang, C. Ding,  An infinite family of linear codes supporting $4$-designs, arXiv:2001.00158, 
accepted for publication in IEEE Trans. Information Theory. 

\bibitem{Tdim} 
V. D. Tonchev, Linear perfect codes and a characterization of the classical designs, 
Designs, Codes and Cryptography 17, 121--128, 1999.

\bibitem{Tonchevhb}
V. D. Tonchev, Codes, in:
Handbook of Combinatorial Designs, 2nd Edition, C. J. Colbourn, and J. H. Dinitz, (Editors), CRC Press, New York, 2007, pp.677--701.

\bibitem{W} E. Witt,  \"Uber Steinersche Systeme,
{\it Abh. Math. Sem. Hamburg}, {\bf 12} (1938), 265 -- 275.

\bibitem{XLW}
C. Xiang, X. Ling, Q. Wang,  Combinatorial $t$-designs from quadratic functions, 
Designs, Codes and Cryptography 88(3), 553-565, 2020.

\bibitem{Z16}
D. V. Zinoviev, The number of Steiner triple systems $S(2^m - 1, 3, 2)$ of rank $2^m -m+2$ over 
$\mathbb{F}_2$, Discrete Math.  339, 2727--2736, 2016.

\bibitem{ZZ12}
V. A. Zinoviev, D. V. Zinoviev, Steiner triple systems $S(2^m - 1, 3, 2)$ of rank $2^m -m+1$ 
over $\mathbb{F}_2$, 
Problems of Information Transmission 48, 102--126, 2012.

\bibitem{ZZ13}
V. A. Zinoviev, D. V. Zinoviev, Structure of Steiner triple systems $S(2^m - 1, 3, 2)$ of rank 
$2^m -m+2$ over $\mathbb{F}_2$, 
Problems of Information Transmission 49, 232--248, 2013.

\bibitem{ZZ13a}
V. A. Zinoviev, D. V. Zinoviev, Remark on ``Steiner triple systems $S(2^m - 1, 3, 2)$ of rank 
$2^m -m+1$ over $\mathbb{F}_2$ published
 in \emph{Probl. Peredachi Inf.}, 2012, no. 2." Problems of Information Transmission 49, 107--111, 2013.








\end{thebibliography}
\end{document}